\documentclass[10pt, final, twocolumn, twoside, romanappendices]{IEEEtran}

\ifCLASSINFOpdf
\else
\fi

\usepackage{graphicx,graphics,epsfig,epstopdf,float}
\graphicspath{{./}}
\usepackage{amsmath} 
\usepackage{amssymb}  
\usepackage{mathrsfs}
\usepackage{enumerate}
\usepackage{subfigure}
\usepackage{colortbl}
\usepackage{color}  
\usepackage{bm}
\usepackage{psfrag}
\usepackage{cite}
\usepackage{algorithm}
\usepackage{algorithmic}
\usepackage{mdwlist}
\usepackage{longtable}
\usepackage{array}
\usepackage{acronym}  

\usepackage{bm}
\usepackage[final]{microtype}  

\DeclareMathAlphabet{\mathsfbr}{OT1}{cmss}{m}{n}
\SetMathAlphabet{\mathsfbr}{bold}{OT1}{cmss}{bx}{n}
\DeclareRobustCommand{\msf}[1]{%
  \ifcat\noexpand#1\relax\msfgreek{#1}\else\mathsfbr{#1}\fi
}

\makeatletter
\newcommand{\msfgreek}[1]{\csname s\expandafter\@gobble\string#1\endcsname}
\makeatother

\DeclareFontEncoding{LGR}{}{} 
\DeclareSymbolFont{sfgreek}{LGR}{cmss}{m}{n}
\SetSymbolFont{sfgreek}{bold}{LGR}{cmss}{bx}{n}
\DeclareMathSymbol{\salpha}{\mathord}{sfgreek}{`a}
\DeclareMathSymbol{\sbeta}{\mathord}{sfgreek}{`b}
\DeclareMathSymbol{\sgamma}{\mathord}{sfgreek}{`g}
\DeclareMathSymbol{\sdelta}{\mathord}{sfgreek}{`d}
\DeclareMathSymbol{\sepsilon}{\mathord}{sfgreek}{`e}
\DeclareMathSymbol{\szeta}{\mathord}{sfgreek}{`z}
\DeclareMathSymbol{\seta}{\mathord}{sfgreek}{`h}
\DeclareMathSymbol{\stheta}{\mathord}{sfgreek}{`j}
\DeclareMathSymbol{\siota}{\mathord}{sfgreek}{`i}
\DeclareMathSymbol{\skappa}{\mathord}{sfgreek}{`k}
\DeclareMathSymbol{\slambda}{\mathord}{sfgreek}{`l}
\DeclareMathSymbol{\smu}{\mathord}{sfgreek}{`m}
\DeclareMathSymbol{\snu}{\mathord}{sfgreek}{`n}
\DeclareMathSymbol{\sxi}{\mathord}{sfgreek}{`x}
\DeclareMathSymbol{\somicron}{\mathord}{sfgreek}{`o}
\DeclareMathSymbol{\spi}{\mathord}{sfgreek}{`p}
\DeclareMathSymbol{\srho}{\mathord}{sfgreek}{`r}
\DeclareMathSymbol{\ssigma}{\mathord}{sfgreek}{`s}
\DeclareMathSymbol{\stau}{\mathord}{sfgreek}{`t}
\DeclareMathSymbol{\supsilon}{\mathord}{sfgreek}{`u}
\DeclareMathSymbol{\sphi}{\mathord}{sfgreek}{`f}
\DeclareMathSymbol{\schi}{\mathord}{sfgreek}{`q}
\DeclareMathSymbol{\spsi}{\mathord}{sfgreek}{`y}
\DeclareMathSymbol{\somega}{\mathord}{sfgreek}{`w}

\DeclareMathSymbol{\svarsigma}{\mathord}{sfgreek}{`c}

\DeclareMathSymbol{\sGamma}{\mathalpha}{sfgreek}{`G}
\DeclareMathSymbol{\sDelta}{\mathalpha}{sfgreek}{`D}
\DeclareMathSymbol{\sTheta}{\mathalpha}{sfgreek}{`J}
\DeclareMathSymbol{\sLambda}{\mathalpha}{sfgreek}{`L}
\DeclareMathSymbol{\sXi}{\mathalpha}{sfgreek}{`X}
\DeclareMathSymbol{\sPi}{\mathalpha}{sfgreek}{`P}
\DeclareMathSymbol{\sSigma}{\mathalpha}{sfgreek}{`S}
\DeclareMathSymbol{\sUpsilon}{\mathalpha}{sfgreek}{`U}
\DeclareMathSymbol{\sPsi}{\mathalpha}{sfgreek}{`Y}
\DeclareMathSymbol{\sOmega}{\mathalpha}{sfgreek}{`W}

\DeclareRobustCommand{\mcal}[1]{%
  \ifcat\noexpand#1\relax\mathnormal{#1}\else\cal{#1}\fi
}
\DeclareRobustCommand{\BM}[1]{%
  \ifcat\noexpand#1\relax\bm{\boldUppercaseItalicGreek{#1}}\else\bm{#1}\fi
}
\makeatletter
\newcommand{\boldUppercaseItalicGreek}[1]{\csname var\expandafter\@gobble\string#1\endcsname}
\makeatother

\definecolor{BLUE}{rgb}{0,0,1}
\newtheorem{theorem}{Theorem}
\newtheorem{proposition}{Proposition}

\newtheorem{remark}{Remark}
\newtheorem{corollary}{Corollary}
\newtheorem{lemma}{Lemma}
\newcommand{\paperTitle}{Data-Aided Target Localization in Multistatic ISAC Systems With Communication Constraints}

\acrodef{gnss}[GNSS]{global navigation satellite system}
\acrodef{rf}[RF]{radio frequency}
\acrodef{aoa}[AOA]{angle-of-arrival}
\acrodef{rss}[RSS]{received signal strength}
\acrodef{toa}[TOA]{time-of-arrival}
\acrodef{tdoa}[TDOA]{time-difference-of-arrival}
\acrodef{rtt}[RTT]{round-trip time}
\acrodef{fdd}[FDD]{frequency division duplex}
\acrodef{tdd}[TDD]{time division duplex}
\acrodef{fd}[FD]{full-duplex}
\acrodef{sdp}[SDP]{semidefinite programming}
\acrodef{psd}[PSD]{positive semi-definite}
\acrodef{crlb}[CRLB]{Cram\'{e}r-Rao lower bound}
\acrodef{nc}[NC]{narrow correlator}
\acrodef{sc}[SC]{storbe correlator}
\acrodef{pll}[PLL]{phase locked loop}
\acrodef{mp}[MP]{multipath}
\acrodef{sp}[SP]{single path}
\acrodef{ff}[FF]{flat fading}
\acrodef{mds}[MDS]{multidimensional scaling}
\acrodef{snr}[SNR]{signal-to-noise ratio}
\acrodef{los}[LOS]{line-of-sight}
\acrodef{nlos}[NLOS]{non-line-of-sight}
\acrodef{sic}[SIC]{serial interference cancelation}
\acrodef{pic}[PIC]{parallel interference cancelation}
\acrodef{adc}[ADC]{analog-to-digital converter}
\acrodef{bp}[BP]{basis pursuit}
\acrodef{lasso}[LASSO]{least absolute shrinkage and selection operator}
\acrodef{omp}[OMP]{orthogonal matching pursuit}
\acrodef{lls}[LLS]{linear least squares}
\acrodef{wlls}[WLLS]{weighted linear least squares}
\acrodef{nlls}[NLLS]{nonlinear least squares}
\acrodef{awgn}[AWGN]{additive white Gaussian noise}
\acrodef{cirf}[CIRF]{channel impulse response function}
\acrodef{irf}[IRF]{impulse response function}
\acrodef{llr}[LLR]{log-likelihood ratio}
\acrodef{llrs}[LLRs]{log-likelihood ratios}
\acrodef{fim}[FIM]{Fisher information matrix}
\acrodef{efim}[EFIM]{equivalent Fisher information matrix}
\acrodef{mse}[MSE]{mean squared error}
\acrodef{speb}[SPEB]{squared position error bound}
\acrodef{peb}[PEB]{position error bound}
\acrodef{rmse}[RMSE]{root mean squared error}
\acrodef{seb}[SEB]{synchronization error bound}
\acrodef{imu}[IMU]{inertial measurement unit}
\acrodef{nls}[NLS]{network localization and synchronization}
\acrodef{Nls}[NLS]{Network localization and synchronization}
\acrodef{reb}[REB]{ranging error bound}
\acrodef{co}[CO]{clock offset}
\acrodef{cdma}[CDMA]{code-division multiple-access}
\acrodef{pdf}[PDF]{probability density function}
\acrodef{isac}[ISAC]{integrated sensing and communication}
\acrodef{ofdm}[OFDM]{orthogonal frequency division multiplexing}
\acrodef{sac}[S\&C]{sensing and communication}
\acrodef{ser}[SER]{symbol error rate}
\acrodef{bp}[BP]{belief propagation}
\acrodef{ula}[ULA]{uniform linear array}
\acrodef{dk}[DK]{deterministic known}
\acrodef{ru}[RU]{random unknown}
\acrodef{tdma}[TDMA]{time division multiplexing access}
\acrodef{map}[MAP]{maximum a posteriori probability}
\acrodef{em}[EM]{expectation-maximization}
\acrodef{mf}[MF]{mean-field}
\acrodef{ep}[EP]{expectation propagation}
\acrodef{mmse}[MMSE]{minimum mean square error}
\acrodef{lmmse}[LMMSE]{linear minimum mean square error}
\acrodef{bpsk}[BPSK]{binary phase shift keying}
\acrodef{qpsk}[QPSK]{quadrature phase shift keying}
\acrodef{qam}[QAM]{quadrature amplitude modulation}
\acrodef{jtd}[JTD]{joint target-localization and data-detection}
\acrodef{csi}[CSI]{channel state information}
\acrodef{bs}[BS]{base station}
\acrodef{aod}[AOD]{angle-of-departure}
\acrodef{ml}[ML]{maximum likelihood}
\acrodef{sinr}[SINR]{signal-to-interference-plus-noise ratio}
\acrodef{wrt}[w.r.t.]{with respect to}
\acrodef{du}[DU]{deterministic unknown}
\acrodef{mi}[MI]{mutual information}
\acrodef{mimo}[MIMO]{multiple-input multiple-output}
\acrodef{siso}[SISO]{single-input single-output}
\acrodef{iid}[i.i.d.]{independent identically distributed}
\acrodef{ls}[LS]{least square}
\acrodef{simo}[SIMO]{single-input multiple-output}
\acrodef{otfs}[OTFS]{orthogonal time frequency space}
\acrodef{bcd}[BCD]{block coordinate descent}
\acrodef{sca}[SCA]{successive convex approximation}
\acrodef{kkt}[KKT]{Karush–Kuhn–Tucker}

\begin{document}
\title{\paperTitle}

\author{
	\vspace{0.2cm}
 Na~Zhao,~\IEEEmembership{Member,~IEEE},
Xiao~Shen,
Chao~Ge,~
Ziping~Lu,~\IEEEmembership{Student~Member,~IEEE}, and
Yuan~Shen,~\IEEEmembership{Senior~Member,~IEEE}


    \thanks{
    N. \ Zhao  is with the Smart City College, Beijing Union University, 100101, China
    (e-mail: {zhaona@buu.edu.cn});
    X. \ Shen, C. \ Ge, Z. \ Lu, and Y. \ Shen  are with the Department of Electronic Engineering, and Tsinghua National Laboratory for Information Science and Technology, Tsinghua University, Beijing 100084, China 
    (e-mail: {shenx.tsinghua@gmail.com}; {gechao@amss.ac.cn}; {lzp23@mails.tsinghua.edu.cn}; {shenyuan\_ee@tsinghua.edu.cn}).
     }

}
\maketitle


\begin{abstract}
\Ac{isac} enables future wireless networks to perform  \ac{sac} over a shared waveform. In multistatic \ac{isac} systems, however, the sensing receivers do not know the realizations of transmitted data symbols, making it challenging to exploit communication signals for sensing. In this paper, we propose a data-aided framework for target localization with two receiver strategies, namely statistical data-aided sensing and joint data-aided sensing and decoding, where the former marginalizes the random unknown data symbols and the latter reuses the reliably decoded data symbols as known virtual pilots. Under \ac{ofdm} signaling, we derive the performance limits for target localization in both strategies and adopt the achievable ergodic data rate as the communication metric. Then, we formulate a joint time-allocation and transmit data-covariance design problem for target localization under communication constraints, which characterizes the joint \ac{sac} bound and quantifies the sensing gain provided by data symbols. In addition, we develop two target localization algorithms that implement the proposed data-aided receiver processing, and extend the framework to finite-alphabet signaling. Simulation results validate theoretical analysis and the effectiveness of the proposed data-aided schemes.
\end{abstract}

\acresetall

%



\acresetall		

\section{Introduction}
\Ac{isac} technology has emerged as a promising enabler for next-generation wireless networks by supporting \ac{sac} over a shared waveform\cite{LiuLiuCui:J25,LiuZhaZha:J24}. This transformative paradigm supports a wide range of applications, including low-altitude economy\cite{TanYuPan:J25}, cooperative localization and navigation\cite{WanWuShe:J20}, and autonomous driving\cite{DuLiuLi:J25}. By sharing hardware platforms and spectral resources, \ac{isac} systems can significantly improve spectral efficiency and reduce deployment costs \cite{LiuMasLiSunHan:J18}. However, the tasks of \ac{sac} use the shared waveform in different ways, where communication relies on \ac{ru} data symbols, whereas sensing benefits from \ac{dk} pilot signals\cite{XioLiuCui:J22}. This mismatch creates an intrinsic coupling between data rate and localization accuracy for \ac{sac}, making transceiver design and theoretical characterization central issues for \ac{isac} system design \cite{XioLiuCui:J22,SheLuZha:J26,KesMojLac:J25}.

Existing studies can be broadly categorized into monostatic and bistatic or multistatic settings \cite{LiuHuaShlLiuZhoEld:J20,ZhaWanZha:J21,DuLiuXio:J24,LiuZhaXio:J25,ZhaChaShe:J25,WuLiHe:J25,KesKoiWym:J21}. In monostatic systems, the transceivers are collocated, allowing the receiver full knowledge of the transmitted waveform, including both codebook and codewords \cite{XioLiuCui:J22}. This property enables transmit-side waveform and beamforming optimization for joint \ac{sac} performance. For instance, \cite{LiuHuaShlLiuZhoEld:J20} proposes a joint multi-beam design to match desired radar beampatterns under communication constraints, while \cite{ZhaWanZha:J21} further develops a joint transmit and receive beamforming scheme accounting for clutter in complex environments. Beyond dedicated waveform design, \cite{DuLiuXio:J24,LiuZhaXio:J25} show that standard \ac{ofdm} signals can also be directly exploited for sensing, for example, by mitigating the effect of random data via probabilistic constellation shaping and utilizing the inherent ranging capability of cyclic prefix-\ac{ofdm}. Building on these transmitter designs, theoretical studies characterize the coupling of \ac{sac} through \ac{crlb}-rate regions and Pareto boundaries, revealing how waveform randomness and deterministic sensing requirements jointly shape the achievable performance\cite{XioLiuCui:J22,HuaHanXu:J24,RenPenSon:J24,GuoGuWan:J25}.

In contrast, bistatic and multistatic systems employ spatially separated transceivers. For target localization, the separated receivers provide geometric diversity in time delay and angle, but each sensing receiver knows only the statistical distribution of random data symbols rather than their realizations \cite{XioLiuCui:J22}. Therefore, a unified receiver-side framework is needed to determine how these unknown data symbols should be processed and how much localization information can be extracted from them under a required data rate. To address data uncertainty, \cite{ZhaChaShe:J25} proposes a joint target-localization and data-detection scheme, where iteratively decoded data symbols progressively improve sensing accuracy. Extending to complex environments, \cite{WuLiHe:J25} develops a variational inference receiver for orthogonal time frequency space \ac{mimo} systems that performs joint data detection and sensing via message passing. For uplink scenarios, \cite{YuRenPan:J25} introduces a projection-type receiver to optimally balance signal-to-inference-plus-noise ratio enhancement and interference suppression, while \cite{MenMasPet:J25} investigates cooperative gains in multistatic networks through multi-point coordinated transmission. These receiver designs demonstrate the potential of data-aided sensing, but the corresponding performance limits under communication constraints remain insufficiently understood.

 \begin{figure}[t]
    \centering
    \includegraphics[width=1\linewidth]{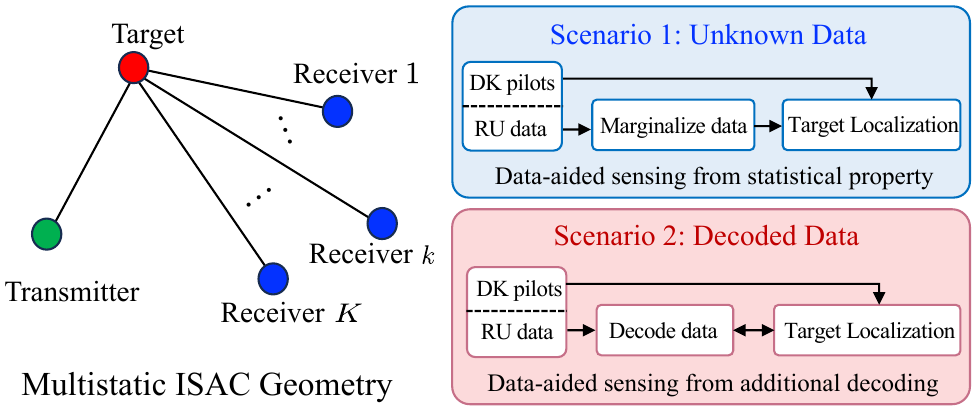}
    \caption{Multistatic \ac{isac} system geometry and data-aided sensing schemes.
Scenario~1 treats \ac{ru} data symbols as nuisance parameters,
whereas Scenario~2 relies on data recovery and exploits recovered data symbols for sensing. DK and RU  refer to deterministic known and random unknown, respectively.}
    \label{fig:system_scheme}
\end{figure}

Several recent works also analyze the theoretical limits of bistatic or multistatic \ac{isac} systems\cite{KwoConParWin:J21}. For instance, \cite{WanWan:J25} identifies waveform uncertainty and spatial water-filling as two structural conflicts, and \cite{SonYuXu:J26} proposes superimposing deterministic sensing sequences onto Gaussian information signals to compensate for waveform uncertainty at the transmitter. Accounting for estimation coupling at the receiver, \cite{FodFodTel:J25} derives the \ac{crlb}s for joint \ac{aoa} and data estimation, showing that stochastic signaling can intensify the \ac{sac} conflict. For delay-constrained bistatic networks, \cite{SheLuZha:J26} characterizes the rate-error tradeoff in the finite-blocklength regime, demonstrating how decoding errors couple with sensing accuracy. These studies reveal the importance of joint theoretical analysis for \ac{sac}, but they do not provide unified performance limits of target localization for different data-aided receiver strategies in multistatic \ac{isac} systems. This leaves open how different receiver processing strategies should exploit unknown data symbols under the same communication constraint, and how much localization information these symbols can contribute.

In particular, the way these unknown data symbols are processed determines the extent to which communication signals can be exploited for sensing, commonly referred to as \emph{data-aided sensing} or communication-assisted sensing. This idea has been extensively studied in communication systems.  In conventional channel estimation, \cite{MaPin:J14} utilizes partially decoded data symbols to iteratively suppress estimation errors. Extending  semi-blind frameworks to massive \ac{mimo} systems, \cite{ZhaKamAlo:J25} establishes theoretical limits by treating unknown data symbols as deterministic parameters, showing that data exploitation can significantly reduce training overhead. Inspired by these advances, recent studies have begun incorporating data-aided processing into \ac{isac} systems.  For monostatic sensing, \cite{XuYuLiu:J25} evaluates monostatic sensing by averaging over data distributions to derive the ergodic linear \ac{mmse}. In bistatic settings, \cite{ZhaChaShe:J25} treats random data symbols as nuisance parameters in the \ac{fim} derivation and approaches this limit through an iterative scheme, as does the data-aided receiver of~\cite{KesMurMiz:C25}. From a more general information-theoretic perspective, \cite{DonLiuLu:J25} models communication-assisted sensing via rate-distortion theory, quantifying how communication capacity constrains sensing parameter recovery. 
Therefore, a unified multistatic localization framework is still needed to quantify how different data-aided receiver processing strategies extract position information from unknown communication symbols under  rate constraints.

In this paper, we establish a unified data-aided target localization framework in multistatic \ac{isac} networks. Under \ac{ofdm} signaling, we consider two receiver processing strategies, namely \emph{statistical data-aided sensing} and
\emph{joint data-aided sensing and decoding}. The former exploits the statistical distribution of the unknown data symbols, whereas the latter decodes the symbols and reuses the reliably  decoded symbols for sensing. Since the two strategies
extract distinct sensing information from the same communication signals, they
give rise to different performance limits for target localization, joint \ac{sac} bounds, and 
algorithm designs.

The main contributions are summarized as follows.
\begin{itemize}
\item We derive the performance limits for target localization in both data-aided strategies, 
quantifying how much sensing information the unknown data symbols provide under different receiver processing strategies.
\item We characterize the joint \ac{sac} bound and jointly optimize time allocation and transmit data covariance, showing that data-aided processing improves the achievable localization accuracy under communication constraints.
\item We develop the target localization algorithms matched to the two strategies and extend the framework to finite-alphabet signaling, which validates the data-aided gains in practical receiver processing.
\end{itemize}

The rest of this paper is organized as follows. 
Section~\ref{sec:sys} presents the system model and problem formulation. 
Section~\ref{sec:data-aided} derives the performance limits for data-aided target localization and communication, and discusses the finite-alphabet extension. 
Section~\ref{sec:optimization} addresses the communication-constrained target localization optimization problem. 
Section~\ref{sec:DAS} develops the target localization algorithms 
Section~\ref{sec:results} provides numerical results, and Section~\ref{sec:concl} concludes the paper.

\emph{Notation:}
Throughout this paper, scalars, vectors, matrices, and sets are denoted by italic letters $x$, bold italic letters $\bm{x}$, bold capital italic letters $\bm{X}$, and calligraphic letters $\mathcal{X}$, respectively.
Random variables, random vectors, and random matrices are denoted by $\mathsf{x}$, bold letters $\bm{\mathsf{x}}$, and bold capital letters $\bm{\mathsf{X}}$, respectively.  $[\bm X]_{i,j}$ denotes an element at the $i$-th row and $j$-th column of matrix $\bm X$.
The operators $(\cdot)^*$, $(\cdot)^{\text T}$, and $(\cdot)^{\text H}$ denote the complex conjugate, transpose, and conjugate transpose, respectively.
The operators $\text{Re}\{\cdot\}$ and $\text{Im}\{\cdot\}$ denote the real and imaginary parts.
The operators $\text{rank}(\cdot)$ and $\text{tr}(\cdot)$ denote the matrix rank and trace.
The symbols $\odot$ and $\otimes$ denote the Hadamard product and Kronecker product.
The notation $\langle\bm X,\bm Y\rangle\triangleq
\text{Re}\{\text{tr}(\bm X^{\text H}\bm Y)\}$ denotes the real Frobenius inner
product.
The operator $\mathbb{E}\{\cdot\}$ denotes expectation.
The relation $\bm{X} \succeq \bm{Y}$ means that $\bm{X} - \bm{Y}$ is positive semidefinite, and $\text{diag}\{\cdot\}$ constructs a diagonal matrix from its input arguments.

\section{System Model and Problem Formulation}\label{sec:sys}
This section introduces the system model for the data-aided schemes, presents the \ac{sac} performance metrics, and formulates the data-aided target localization problem. 

\subsection{System Model} \label{sec:sysmod}
Consider an \ac{ofdm}-based $2\text{D}$  multistatic \ac{isac} system, consisting of one transmitter, one point target, and $K$ receivers, as illustrated in Fig.~\ref{fig:system_scheme}. The transmitter sends common data symbols to all receivers, which jointly decode them and localize the target. The positions of the transmitter and the $k$-th receiver are assumed to be known, denoted by $\bm p_{\text{t}} = [x_{\text{t}},y_{\text{t}}]^\text{T}$ and $\bm p_{\text{r},k} = [x_{\text{r},k},y_{\text{r},k}]^\text{T}$, $k\in\mathcal{K} \triangleq \{1,\cdots,K\}$, respectively. The target position is denoted by $\bm p = [x,y]^\text{T}$, which is the parameter of interest to be estimated. The system operates over $N$ orthogonal subcarriers indexed by $\mathcal N\triangleq\{0,1,\ldots,N-1\}$,
with subcarrier spacing $\Delta f$ and total bandwidth
$B=N\Delta f$.

\subsubsection{Transmitted Signal Model}
Following the \ac{tdd} protocol, each transmission frame is divided into pilot and data phases, in which pilot phase occupies $T_{\text p}=\rho T$
time slots, while the remaining $T_{\text d}=(1-\rho)T$
time slots are used for data transmission, where $\rho\in[0,1]$ denotes the pilot fraction controlling the \ac{sac} resource allocation and $T$ denotes the total time slots.\footnote{Without loss of generality, we assume  $\rho T$ is an integer throughout the analysis, such that
$T_{\text p}=\rho T$ and $T_{\text d}=(1-\rho)T$. The rounding discrepancy is of order $\mathcal O(1/T)$ and becomes negligible for practical frame lengths.} The pilot and data index sets are defined as $\mathcal T_{\text p} \triangleq \{1,\ldots,T_{\text p}\}$, and $\mathcal T_{\text d} \triangleq \{T_{\text p}+1,\ldots,T\}$, respectively. The transmitted signal is defined at the per-subcarrier level to provide a common description for both the \ac{mimo} and \ac{siso} models. At subcarrier $n \in \mathcal{N}$ and time slot $t \in \mathcal{T}\triangleq\{1,\ldots,T\}$, the transmitted signal $\bm{\mathsf{s}}[n,t]\in\mathbb C^{M_{\text t}}$ is given by
\begin{equation}
   \bm{\mathsf{s}}[n,t]=
    \begin{cases}
        \bm{s}_{\text p}[n,t], & t\in\mathcal T_{\text p},\\
         \bm{\mathsf{s}}_{\text d}[n,t], & t\in\mathcal T_{\text d},
    \end{cases}
    \label{eq:unified_signal}
\end{equation}
where $\bm{s}_{\text p}[n,t]$ and $\bm{\mathsf{s}}_{\text d}[n,t]$ denote the pilot and data symbols, respectively. The \ac{siso} case is obtained by setting
$M_{\text t}=1$. The pilot symbols are \ac{dk} at the receivers. We adopt spatially orthogonal pilots on every subcarrier. Defining $\bm S_{{\text p},n}\triangleq[\bm s_{\text p}[n,1],\ldots,\bm s_{\text p}[n,T_{\text p}]]$, their per-subcarrier sample covariance satisfies\cite{XioLiuCui:J22}
\begin{equation}
   \bm R_{\text p}= {T_{\text p}^{-1}}\bm S_{{\text p},n}\bm S_{{\text p},n}^{\text H}
    ={M_{\text t}^{-1}}\bm I_{M_{\text t}},
    \qquad n\in\mathcal N.
    \label{eq:orthogonal_pilot_main}
\end{equation}
The \ac{ru} data symbols are drawn independently across \ac{ofdm} subcarriers and time slots from a zero-mean Gaussian codebook.\footnote{Gaussian signaling is adopted for analytical tractability. This assumption does not rule out practical modulation schemes, whose finite-alphabet extensions are discussed separately.} Their per-subcarrier statistical covariance is 
$  \bm R_{\text d}\triangleq\mathbb E\{\bm{\mathsf s}_{\text d}[n,t]\bm{\mathsf s}_{\text d}^{\text H}[n,t]\}$,
where $R_{\text d}=1$ in the \ac{siso} case and $\bm R_{\text d}\succeq\bm0$, $\text{tr}(\bm R_{\text d})=1$ in the \ac{mimo} case. 
Without loss of generality, the  transmit power is normalized so that $\text{tr}(\bm R_\text{p})=\text{tr}(\bm R_\text{d})=1$. 
Stacking  over the $N$ subcarriers as $\tilde{\bm{\mathsf{s}}}[t] \triangleq [\bm{\mathsf{s}}^{\text T}[0,t], \ldots, \bm{\mathsf{s}}^{\text T}[N-1,t]]^{\text T} \in \mathbb{C}^{NM_{\text t}}$ yields the transmitted signal matrix as
\begin{equation}
    \bm{\mathsf S}
    =
    [\bm S_{\text p},\bm{\mathsf S}_{\text d}]
    \in\mathbb C^{NM_{\text t}\times T},
    \label{eq:signal_matrix}
\end{equation}
with a deterministic pilot block $\bm S_{\text p} \in \mathbb{C}^{NM_{\text t} \times T_{\text p}}$ and a random data block $\bm{\mathsf S}_{\text d} \in \mathbb{C}^{NM_{\text t} \times T_{\text d}}$. Condition~\eqref{eq:orthogonal_pilot_main} constrains the diagonal subcarrier blocks of $T_{\text p}^{-1}\bm S_{\text p}\bm S_{\text p}^{\text H}$, and it does not require the stacked $NM_{\text t}$ dimensions to be mutually orthogonal. The stacked data covariance is
$\tilde{\bm R}_{\text d}=\bm I_N\otimes\bm R_{\text d}$.

\subsubsection{MIMO Channel Model}
The transmitter and receiver are equipped with \acp{ula} of $M_{\text t}$ and
$M_{\text r}$ antennas, respectively. 
The received signal on the $n$-th subcarrier at the $k$-th receiver and the
$t$-th time slot is modeled as\cite{LiCosMit:J22}
\begin{equation}
    \bm{\mathsf y}_k[n,t]
    =
    \alpha_k
    e^{-j2\pi n\Delta f\tau_k}
    \bm A(\bm\vartheta_k)
    \bm{\mathsf s}[n,t]
    +
    \bm{\mathsf v}_k[n,t],
    \label{eq:mimo_subcarrier_model}
\end{equation}
where $ \bm A(\bm\vartheta_k) \triangleq  \bm a_{\text r}(\varphi_k) \bm a_{\text t}^{\text H}(\psi) \in \mathbb C^{M_{\text r}\times M_{\text t}}$ denotes the spatial array response, $\varphi_k$ denotes the \ac{aoa} of the target-reflected signal at the $k$-th receiver, $\psi$ denotes the \ac{aod} from the transmitter toward the target, and $\bm\vartheta_k \triangleq [\varphi_k,\psi]^{\text T}$ collects the angular parameters.\footnote{We focus on the target-reflected component that couples communication and localization. The direct path and static background can be separately calibrated or subtracted~\cite{ZhaLiLiuHim:J15}, with residual mismatch absorbed into the noise.}  Note that $\psi$ is shared by all receivers.
The observed noise is denoted by 
$\bm{\mathsf v}_k[n,t] \sim \mathcal{CN} (\bm 0,\sigma_v^2\bm I_{M_{\text r}})$.
Define the frequency-domain delay steering vector and its corresponding
diagonal matrix as
$\bm g(\tau_k)= [1,e^{-j2\pi \Delta f \tau_k},\cdots,
e^{-j2\pi (N-1)\Delta f \tau_k}]^{\text T}$ and
$\bm D(\tau_k)\triangleq\text{diag}(\bm g(\tau_k))\in\mathbb C^{N\times N}$,
respectively. Then, stacking all subcarriers yields
\begin{equation}
    \tilde{\bm{\mathsf y}}_k[t]
    =
    \alpha_k
    \big(
    \bm D(\tau_k)\otimes\bm A(\bm\vartheta_k)
    \big)
    \tilde{\bm{\mathsf s}}[t]
    +
    \tilde{\bm{\mathsf v}}_k[t],
    \label{eq:mimo_vector_model}
\end{equation}
where the collective received signal and noise are denoted by
$\tilde{\bm{\mathsf y}}_k[t]
=
[\bm{\mathsf y}_k[0,t]^{\text T},
\ldots,
\bm{\mathsf y}_k[N-1,t]^{\text T}]^{\text T}
\in
\mathbb C^{NM_{\text r}\times1}$
and
$\tilde{\bm{\mathsf v}}_k[t]
\sim
\mathcal{CN}
(
\bm 0,
\sigma_v^2\bm I_{NM_{\text r}}
)$. Finally, stacking over all time slots gives
\begin{equation}
    \tilde{\bm{\mathsf Y}}_k
    =
    \alpha_k
    \big(
    \bm D(\tau_k)\otimes\bm A(\bm\vartheta_k)
    \big)
    \bm{\mathsf S}
    +
    \tilde{\bm{\mathsf V}}_k,
    \label{eq:mimo_matrix_model}
\end{equation}
where $\tilde{\bm{\mathsf Y}}_k \in \mathbb{C}^{NM_{\text r}\times T}$ and $\tilde{\bm{\mathsf V}}_k \in \mathbb{C}^{NM_{\text r}\times T}$, partitioned as $\tilde{\bm{\mathsf Y}}_k = [\tilde{\bm{\mathsf Y}}_{\text{p},k},\tilde{\bm{\mathsf Y}}_{\text{d},k}]$ and $\tilde{\bm{\mathsf V}}_k = [\tilde{\bm{\mathsf V}}_{\text{p},k},\tilde{\bm{\mathsf V}}_{\text{d},k}]$, consistent with~\eqref{eq:signal_matrix}.
\begin{remark}
The \ac{siso} model is recovered as a special case by setting
$M_{\text t}=M_{\text r}=1$. The signal $\mathsf s[n,t]$ on each subcarrier
then becomes scalar. For conventional \ac{ofdm} with Gaussian signaling, the
data symbols are independent across subcarriers and have unit variance. The
received signal at $k$-th receiver is therefore
$\bm{\mathsf y}_k[t] = \alpha_k\bm g(\tau_k)\odot\bm{\mathsf s}[t]+\bm{\mathsf v}_k[t]$ or stacking over all time slots as
$ \bm{\mathsf Y}_k =\alpha_k\bm D(\tau_k)
\bm{\mathsf S}+\bm{\mathsf V}_k$.
\end{remark}

\subsubsection{Data-aided Schemes}
Since the sensing receiver does not have access to the realizations of the transmitted data symbols and only knows their statistical distribution, the receiver design fundamentally depends on how these unknown data symbols are processed. Accordingly, two data-aided receiver processing strategies are considered, as summarized in Fig.~\ref{fig:system_scheme}. 
\begin{itemize}
    \item In \emph{Scenario~1 (statistical data-aided sensing)}, the unknown data symbols $\bm{\mathsf s}_{\text d}$ are treated as nuisance random variables and marginalized to obtain the marginal likelihood $f(\bm{\mathsf y}; \bm p) = \int f(\bm{\mathsf y}|\bm{\mathsf s}_{\text d}, \bm p)f(\bm{\mathsf s}_{\text d})\,d\bm{\mathsf s}_{\text d}$, enabling data-aided target localization without explicit data decoding. 
    \item In \emph{Scenario~2 (joint data-aided sensing and decoding)}, the target position $\bm p$ and the data symbols $\bm{\mathsf s}_{\text d}$ are jointly estimated by exploiting prior information, e.g., codebook constraints, through the joint model $f(\bm{\mathsf y}, \bm{\mathsf s}_{\text d}; \bm p) = f(\bm{\mathsf y}|\bm{\mathsf s}_{\text d}, \bm p) f(\bm{\mathsf s}_{\text d})$, enabling data-aided target localization by reusing reliably recovered data symbols as virtual pilots.
\end{itemize}
	These two data-aided schemes show that the amount and nature of sensing information extractable from data are  determined by the receiver-side processing strategy. 

\subsection{Performance Metrics}
To evaluate the proposed data-aided target localization framework under communication constraints, we adopt the \ac{speb} as the sensing metric and the achievable ergodic data rate as the communication metric, respectively~\cite{XioLiuCui:J22,SheLuZha:J26}. This provides a unified performance formulation applicable to both \ac{mimo} and \ac{siso} systems.

\subsubsection{Sensing Performance}
The sensing objective is to estimate the target position $\bm p$, which is coupled with the unknown nuisance complex amplitude vector
$\bm \kappa = [\bm \alpha_1^\text{T},\cdots,\bm \alpha_{K}^\text{T}]^\text{T}$,
\label{eq:nuisance_def}
where
$\bm {\alpha}_k = [\text{Re}\{\alpha_k\},\text{Im}\{\alpha_k\}]^\text{T}$.
Accordingly, the unknown parameter vector is defined as
$\bm\eta=[\bm p^\text{T},\bm\kappa^\text{T}]^\text{T}$.

For a complex Gaussian observation with mean $\bm\mu(\bm\eta)$ and covariance $\bm Q(\bm\eta)$, the $(i,j)$-th entry of the \ac{fim} is given by~\cite{KhaPodHaa:J21}
\begin{equation}
    [\bm J(\bm\eta)]_{i,j}
    =
    2\text{Re}\Big\{
    \frac{\partial\bm\mu^\text H}{\partial\eta_i}
    \bm Q^{-1}
    \frac{\partial\bm\mu}{\partial\eta_j}
    \Big\}
    +
    \text{tr}\Big(
    \bm Q^{-1}\frac{\partial\bm Q}{\partial\eta_i}
    \bm Q^{-1}\frac{\partial\bm Q}{\partial\eta_j}
    \Big).
    \label{eq:direct_gaussian_fim}
\end{equation}
Based on the per-subcarrier signal model~\eqref{eq:mimo_subcarrier_model}, the two
scenarios differ in how the data component contributes to the Fisher information. In Scenario~1, the pilots are known and contribute through the mean-dependent term, while the unknown Gaussian data are marginalized and contribute only through the covariance-dependent term. 
In Scenario~2, correctly decoded data symbols are treated as a known waveform and therefore contribute through the same mean-dependent mechanism as the pilots.\footnote{Scenario~2 serves as a reliable-recovery benchmark, characterizing
the sensing information available once the data block has been recovered
accurately. This regime is motivated by joint data-detection and data-aided
sensing receivers, where outside the low-\ac{snr} region most recovered symbols
are correct and the recovered data block can be approximated as known for
sensing~\cite{ZhaChaShe:J25,KesMurMiz:C25}.} In Scenario~1, the unknown data symbols are marginalized, and hence their contribution to the \ac{fim} is characterized through the covariance
of the received signal in \eqref{eq:mimo_subcarrier_model}, calculated by
\begin{equation}
    \tilde{\bm Q}_k
    \triangleq
    \mathbb E\left\{
        \bm{\mathsf y}_k[n,t]\bm{\mathsf y}_k^{\text H}[n,t]
    \right\}
    =
    \bm H_k\bm R_{\text d}\bm H_k^{\text H}
    +\sigma_v^2\bm I_{M_{\text r}},~t\in \mathcal{T}_{\text{d}},
    \label{eq:Qk_def}
\end{equation}
where $\bm H_k=\alpha_k\bm A(\bm\vartheta_k)$ denotes the sensing spatial channel. Since $\bm H_k$ depends on the target position through the angular parameters, the covariance $\tilde{\bm Q}_k$ preserves spatial information about the target.
The corresponding \ac{siso} covariance  $\tilde Q_k=|\alpha_k|^2+\sigma_v^2$, which is a scalar independent of the target position, so marginalized data provides no direct position information.

To derive the performance limits for target localization, we introduce the intermediate geometric channel parameter vector
$\bm \xi =[\bm \tau^\text{T},\bm \varphi^\text{T},\psi,
\bm \kappa^\text{T}]^\text{T}$ for \ac{mimo} case, where $\bm \tau = [\tau_1,\cdots,\tau_K]^\text{T}$ and $\bm \varphi = [\varphi_1,\cdots,\varphi_K]^\text{T}$, 
which reduces to
$\bm \xi = [\bm \tau^\text{T},\bm \kappa^\text{T}]^\text{T}$ for \ac{siso} case.
The \ac{fim} \ac{wrt} $\bm\eta$ is obtained through the Jacobian transformation $\bm J(\bm\eta)=\bm T_\eta^\text{T}\bm J(\bm \xi)\bm T_\eta$,
where $\bm T_\eta \triangleq
{\partial \bm \xi}/{\partial\bm\eta^\text{T}}$. Partitioning the \ac{fim} as
\begin{equation}
    \bm J(\bm\eta)=
    \begin{bmatrix}
        \bm J_{\bm p\bm p} & \bm J_{\bm p \bm \kappa}\\
        \bm J_{\bm \kappa \bm p} & \bm J_{\bm \kappa \bm \kappa}
    \end{bmatrix},
\end{equation}
the \ac{efim} for the target position is given by 
 $\bm J_{\text e}(\bm p;\bm x)
    =
    \bm J_{\bm p\bm p}
    -
    \bm J_{\bm p \bm \kappa}
    \bm J_{\bm \kappa\bm \kappa}^{-1}
    \bm J_{\bm \kappa \bm p}$\cite{SheWin:J10a}.
Let $\hat{\bm p}$ be an unbiased estimates of the target position, then the \ac{mse} matrix of $\hat{\bm p}$
satisfies
\begin{equation}
    \mathbb E\left\{
    (\hat{\bm p}-\bm p)(\hat{\bm p}-\bm p)^\text{T}
    \right\}
    \succeq
    \bm J_{\text e}^{-1}(\bm p;\bm x),
\end{equation}
where the design variable is $\bm x=\rho$ for \ac{siso} time-allocation
design and $\bm x=(\rho,\bm R_{\text d})$ for \ac{mimo} joint
time-allocation and transmit-covariance design. 
Then, the \ac{mse} of the position estimate is bounded below by the \ac{speb}, defined
as\cite{SheWin:J10a}
\begin{equation}
    \mathcal{P}(\bm p;\bm x)
    \triangleq
    \text{tr}\left\{
    \bm J_{\text e}^{-1}(\bm p;\bm x)
    \right\}.
    \label{eq:speb_metric}
\end{equation}

\subsubsection{Communication Performance}
The communication performance is characterized by the achievable ergodic rate under pilot-based channel estimation. For receiver~$k$, the achievable rate $\bar{R}_k(\bm x)$ accounts for both the data-phase overhead $(1-\rho)$ and the channel estimation error with design variable $\bm x$. Since all receivers decode the same data symbols, the common-message broadcast rate is\cite{RenPenSon:J24}
\begin{equation}
    \bar{R}(\bm x) = \min_{k=1,\ldots,K}\; \bar{R}_k(\bm x).
    \label{eq:system_rate}
\end{equation}
The detailed rate expressions are derived in Section~\ref{sec:comm}.

\subsubsection{Joint Bound and Problem Formulation}
To jointly characterize \ac{sac} performance, we use the achievable data rate
and the localization bound as two coupled metrics. For visualization, the joint
\ac{sac} performance set is defined as
\begin{equation}
\mathcal C
 =
\left\{
(r,\varepsilon):
\exists\, \bm x \in \mathcal X ~~\text{s.t.}~
0\le r \le \bar R(\bm x),
~
\varepsilon \ge \mathcal{P}(\bm p;\bm x)
\right\},
\label{eq:joint_sac_set}
\end{equation}
where $\mathcal X=\mathcal X_{\text S}$ for \ac{siso} and
$\mathcal X=\mathcal X_{\text M}$ for \ac{mimo}, with
$\mathcal X_{\text S}=[\rho_{\min},1]$ and
$\mathcal X_{\text M}=\{(\rho,\bm R_{\text d}):
\rho_{\min}\le\rho\le1,\bm R_{\text d}\succeq\bm0,
\operatorname{tr}(\bm R_{\text d})=1\}$. The lower bound
$\rho_{\min}=M_{\text t}/T$ is imposed by the per-subcarrier spatially
orthogonal pilot condition.

For a prescribed communication requirement, the corresponding target localization
performance is obtained by solving the constrained optimization problem, given by
\begin{equation}
\min_{\bm x \in \mathcal X}\quad
\mathcal{P}(\bm p;\bm x),~~~~~~~\text{s.t.}~
\bar R(\bm x) \ge R_\text{th},
\label{eq:constrained_bound_problem}
\end{equation}
where $R_\text{th}$ is the minimum required  data rate. Varying
$R_\text{th}$ traces the joint \ac{sac} bound under communication constraints.

\section{Performance Limits Analysis for Data-aided Multistatic ISAC Systems}  \label{sec:data-aided}
In this section, we derive the performance limits for target localization under two data-aided schemes, and then characterize the ergodic rate for communication. Besides, we discuss the extensions from Gaussian signaling to finite-alphabet signaling.

\subsection{Geometric Channel Parameters}
\label{subsec:geo_params}
For later derivations, this subsection first establishes the geometric mapping
from the target position $\bm p$ to the intermediate geometric channel parameters
$\bm\theta_k=[\tau_k,\varphi_k,\psi]^{\text T}$ for the $k$-th bistatic link,
including bistatic delay $\tau_k$,  \ac{aoa} $\varphi_k$, and  \ac{aod}
$\psi$. The derivatives of $\bm \theta_k$ \ac{wrt}
 target position $\bm p$ are then derived to identify the information directions
associated with different measurements.
For $k$-th bistatic link, the geometric relationship between channel parameters and target position is given by 
\begin{equation}
    \begin{aligned}
        &~~~~~~\tau_k  = \frac{1}{c}(\Vert\bm p_{\text{t}}-\bm p\Vert +\Vert\bm p_{\text{r},k}-\bm p\Vert), \\
        &\psi = \text{arctan}\frac{y_{\text{t}}-y}{x_{\text{t}}-x},\quad\varphi_k = \text{arctan}\frac{y_{\text{r},k}-y}{x_{\text{r},k}-x},
    \end{aligned}
\end{equation}
where $c$ is the speed of light. The corresponding unit direction vectors for \ac{aod} and \ac{aoa} are $\bm e_{\text{t}} = [\cos\psi, \sin\psi]^\text{T}$ and $\bm e_{\text{r},k} = [\cos\varphi_k, \sin\varphi_k]^\text{T}$, with distances $d_{\text{t}} = \|\bm p - \bm p_{\text{t}}\|$  and $d_{\text{r},k} = \|\bm p - \bm p_{\text{r},k}\|$, respectively. 
The bistatic time-delay gradient \ac{wrt} $\bm p$ admits the bisector decomposition\cite{SheDaiWin:J14}
\begin{equation}
    \frac{\partial \tau_k}{\partial \bm p}
    = \frac{2}{c}\cos\Big(\frac{\psi - \varphi_k}{2}\Big)
    \bar{\bm q}_k,
    \quad
    \bar{\bm q}_k \triangleq \begin{bmatrix} \cos\bar\phi_k \\ \sin\bar\phi_k \end{bmatrix},
    \label{eq:bisector_decomp}
\end{equation}
where $\bar\phi_k = (\psi + \varphi_k)/2$ is the bistatic bisector angle. Besides, for \ac{mimo} systems, the angle gradients \ac{wrt} $\bm p$ are
\begin{equation}
    \frac{\partial \varphi_k}{\partial \bm p}
    = \frac{1}{d_{\text{r},k}} \bm e_{\text{r},k}^{\perp},
    \quad
    \frac{\partial \psi}{\partial \bm p}
    = \frac{1}{d_{\text{t}}} \bm e_{\text{t}}^{\perp},
    \label{eq:angle_gradients}
\end{equation}
where $\bm e^{\perp}$ denotes the $90^\circ$ counterclockwise rotation of $\bm e$.
Since the
complex amplitudes are independent of the target position, the nuisance amplitudes are first eliminated in the channel domain, yielding the \ac{efim}
$\bm J_{\text e}(\bm\theta)$ of the geometric channel parameters. The \ac{efim}
of target localization is then obtained through  geometric Jacobian matrix $\bm T$, given by
$\bm J_{\text e}(\bm p)
    =
    \bm T^{\text T}\bm J_{\text{e}}(\bm \theta)\bm T,
    \bm T
    =
    {\partial \bm \theta}/{\partial \bm p}$,
where $\bm \theta=[\bm \theta_1^\text{T},\ldots,\bm \theta_K^\text{T}]^\text{T}$.
Using the gradients in \eqref{eq:bisector_decomp} and
\eqref{eq:angle_gradients}, the \ac{efim} of target localization can be written as
\cite{SheWin:J10a}
\begin{equation}
    \bm J_{\text e}(\bm p)
    =
    \sum_{k=1}^{K}\Big(
    \mu_k^{\tau}\, \bar{\bm q}_k \bar{\bm q}_k^\text{T}
    + \mu_k^{\varphi}\, \bm e_{\text{r},k}^{\perp} {\bm e_{\text{r},k}^{\perp}}^\text{T}
    \Big)
    + \mu^{\psi}\, \bm e_{\text{t}}^{\perp} {\bm e_{\text{t}}^{\perp}}^\text{T},
    \label{eq:unified_efim}
\end{equation}
where $\mu^{\psi}=\sum_{k=1}^{K}\mu_k^{\psi}$, and the information intensities
$\mu_k^{\tau}$, $\mu_k^{\varphi}$, and $\mu_k^{\psi}$ depend on the system
configuration and scenario. For \ac{siso} case, only the time delay term is present with
$\mu_k^{\varphi}=\mu^{\psi}=0$. For \ac{mimo} case, the three terms carry information
along distinct spatial directions, providing inherent angular diversity even with
a single receiver.

\begin{remark}
\label{rem:info_directions}
The decomposition in~\eqref{eq:unified_efim} separates the topology-dependent
directions $\bar{\bm q}_k$, $\bm e_{\text{r},k}^{\perp}$, and
$\bm e_{\text{t}}^{\perp}$, which are fixed by the network geometry, from the
waveform-dependent intensities $\mu_k^{\tau}$, $\mu_k^{\varphi}$, and $\mu^{\psi}$.
Pilot allocation, data-covariance design, and data-aided processing change only
these scalar intensities, so optimizing the \ac{speb} amounts to shaping the
intensities along pre-determined spatial directions. This structure simplifies
the construction of the communication-constrained sensing objective.
\end{remark}

\begin{figure}[t]
    \centering
    \includegraphics[width=0.9\linewidth]{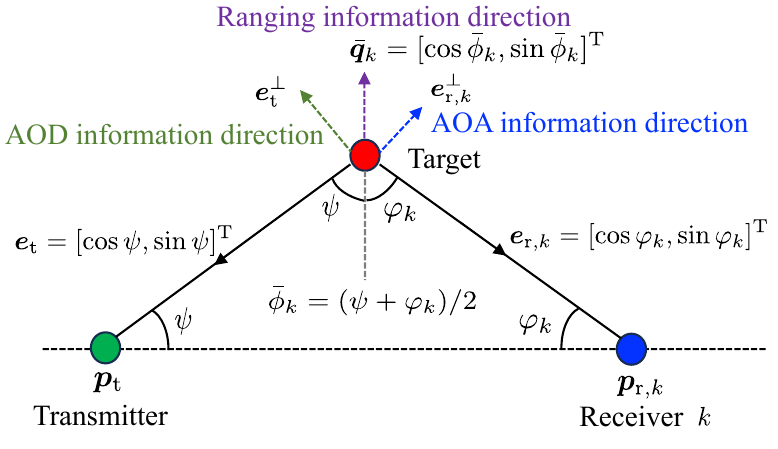}
    \caption{Spatial information directions for multistatic target localization: the bistatic bisector $\bar{\bm q}_k$ for delay, the receiver perpendicular $\bm e_{\text{r},k}^{\perp}$ for \ac{aoa}, and the transmitter perpendicular $\bm e_{\text{t}}^{\perp}$ for \ac{aod}. \ac{siso} retains only $\bar{\bm q}_k$.}
    \label{fig:info_directions}
\end{figure}

\subsection{Performance Limits for Target Localization}
\label{subsec:loc}
To quantify the data-aided gain, we first derive the performance limits for
target localization following 
\eqref{eq:direct_gaussian_fim}, i.e., known pilots provide the mean-information
baseline, marginalized Gaussian data in Scenario~1 provide covariance
information, and decoded data in Scenario~2 provide additional mean
information. For each link,
define the per-link augmented parameter vector as
$\bm\xi_k\triangleq
[\bm\theta_k^{\text T},\bm\alpha_k^{\text T}]^{\text T}$.
The channel-domain \ac{efim} of $\bm\theta_k = [\tau_k,\varphi_k,\psi]^\text{T}$ is obtained by eliminating the
nuisance amplitude $\bm\alpha_k = [\alpha_{\text{R},k},\alpha_{\text{I,k}}]^\text{T}$ with $\alpha_{\text{R},k} = \text{Re}\{\alpha_k\}$ and $\alpha_{\text{I},k} = \text{Im}\{\alpha_k\}$ from the \ac{fim} of $\bm\xi_k$. 

\subsubsection{Deterministic Known Pilot Part}
During the pilot phase, the transmitted symbols is known at the receiver, so
the vectorized observation has a parameter-dependent mean
$\alpha_k\text{vec}((\bm D(\tau_k)\otimes\bm A(\bm\vartheta_k))
\bm S_{\text p})$ and a parameter-independent noise covariance
$\sigma_v^2\bm I_{M_{\text r}NT_{\text p}}$. Thus, the pilot contribution for estimating channel parameters $\bm \theta_k$ is a
mean-information term, with the unknown complex amplitude treated as the
nuisance parameter. The \ac{efim} for $\bm \theta_k$ is given in the following Theorem. 

\begin{theorem}
\label{thm:mimo_s1_pilot}
Under the orthogonal pilot condition in~\eqref{eq:orthogonal_pilot_main}, the \ac{efim} of channel domain intermediate parameter $\bm \theta_k$ is 
\begin{equation}
\begin{aligned}
&\bm J_{\text p}^{\text e}(\bm\theta_k)
=\frac{2T_{\text p}|\alpha_k|^2}{\sigma_v^2}
\text{diag}\Big\{
M_{\text r}\omega^2
\big[s_2(N)-\frac{s_1^2(N)}{N}\big],\\
&~~~~~N\nu_{\text r}^2
\big[s_2(M_{\text r})-\frac{s_1^2(M_{\text r})}{M_{\text r}}\big],\frac{NM_{\text r}}{M_{\text t}}\nu_{\text t}^2
\big[s_2(M_{\text t})-\frac{s_1^2(M_{\text t})}{M_{\text t}}\big]
\Big\}
\end{aligned}
    \label{eq:mimo_pilot_diag}
\end{equation}
where 
$s_\ell(L)\triangleq\sum_{m=0}^{L-1}m^\ell$ for $\ell\in\{1,2\}$,
$\nu_{\text r}\triangleq u_{\text r}'(\varphi_k)$,
$\nu_{\text t}\triangleq u_{\text t}'(\psi)$, and
$\omega\triangleq2\pi\Delta f$. Here, $L$ denotes the relevant aperture
dimension instantiated as $N$, $M_{\text r}$, or $M_{\text t}$, while
$u_q'(\theta_q)$ is the derivative of the inter-element spatial phase of array
$q\in\{\text r,\text t\}$ with respect to the angle, as defined explicitly in
Appendix~\ref{app:proof_pilot_fim}.
\end{theorem}

\emph{Proof:} See Appendix~\ref{app:proof_pilot_fim}.

\begin{remark}
Theorem~\ref{thm:mimo_s1_pilot} reveals two structural properties of the
pilot-induced channel-domain \ac{efim}. First, since
$s_2(L)-s_1^2(L)/L=\sum_{m=0}^{L-1}(m-(L-1)/2)^2$, only the centered frequency
and array apertures carry information, because the common phase is absorbed by
the unknown amplitude $\alpha_k$ and only relative phase variations across
subcarriers or antennas remain informative. Second, orthogonal pilots cancel the
cross terms and decouple the delay, \ac{aoa}, and \ac{aod} contributions, whose
intensities scale linearly with $T_{\text p}$ and the link \ac{snr}. This pilot
contribution provides the common non-data-aided baseline for both data-aided
scenarios.
\end{remark}

\subsubsection{Random Unknown Data Part in Scenario 1}
Scenario~1 refers to statistical data-aided sensing, in which the receiver knows the Gaussian distribution of data symbols but not their realizations. The data observation in
$\tilde{\bm{\mathsf y}}_{\text d,k}[t]
=\alpha_k(\bm D(\tau_k)\otimes\bm A(\bm\vartheta_k))\tilde{\bm{\mathsf s}}_{\text d}[t]
+\tilde{\bm{\mathsf v}}_k[t]$ is therefore averaged out using
$\tilde{\bm{\mathsf s}}_{\text d}[t]\sim
\mathcal{CN}(\bm0,\bm I_N\otimes\bm R_{\text d})$. Unlike the pilot part, the
resulting likelihood for data part in Scenario 1 has zero mean and carries information only through its
covariance. 

\begin{lemma}
\label{thm:mimo_s1_data}
Under Gaussian signaling for data symbols in Scenario 1, the \ac{fim} of all unknown parameters $\bm \xi_k$, which can be partitioned as $(\tau_k,\bm\vartheta_k,\bm\alpha_k)$, has  the block form
\begin{equation}
    \bm J_{\text d}(\bm \xi_k)
    =
    T_{\text d}
    \begin{bmatrix}
    0 & \bm0_{1\times2} & \bm0_{1\times2}\\
    \bm0_{2\times1} & \bm\Gamma_{\text d,k}
    & \bm b_{\text d,k}\bm \alpha_k^{\text T}\\
    \bm0_{2\times1} & \bm \alpha_k\bm b_{\text d,k}^{\text T}
    & \zeta_k\bm \alpha_k\bm \alpha_k^{\text T}
    \end{bmatrix},
    \label{eq:main_data_theta_fim}
\end{equation}
where the covariance matrix $\tilde{\bm Q}_k$ is defined in \eqref{eq:Qk_def}, and
\begin{equation}
\begin{aligned}
    &[\bm\Gamma_{\text d,k}]_{i,j}
    =
    \langle
    \dot{\tilde{\bm Q}}_{i,k},
    \dot{\tilde{\bm Q}}_{j,k}
    \rangle_k,
    ~~ i,j\in\{\varphi,\psi\},\\
    &[\bm b_{\text d,k}]_i
    =
    2\langle\dot{\tilde{\bm Q}}_{i,k},\bm C_k\rangle_k,
    ~~~
    \zeta_k=4\langle\bm C_k,\bm C_k\rangle_k,
    ~~ i\in\{\varphi,\psi\},
\end{aligned}
    \label{eq:main_data_gamma_def}
\end{equation}
with $\dot{\tilde{\bm Q}}_{\varphi,k}\triangleq
\partial\tilde{\bm Q}_k/\partial\varphi_k$,
$\dot{\tilde{\bm Q}}_{\psi,k}\triangleq
\partial\tilde{\bm Q}_k/\partial\psi$, the covariance inner product
$\langle\bm X,\bm Y\rangle_k\triangleq
N\text{tr}(\tilde{\bm Q}_k^{-1}\bm X\tilde{\bm Q}_k^{-1}\bm Y)$,
$\bm C_k=\bm A(\bm\vartheta_k)\bm R_{\text d}\bm A^{\text H}(\bm\vartheta_k)$, and
$\bm \alpha_k=[\alpha_{\text R,k},\alpha_{\text I,k}]^{\text T}$.
\end{lemma}

\emph{Proof:} See Appendix~\ref{app:proof_data_fim}.

\begin{remark}
Lemma~\ref{thm:mimo_s1_data} reveals that statistical data-aided sensing is
 covariance-based. Marginalization removes the delay-dependent
subcarrier phases, so unknown data provide no delay information. In
contrast, \ac{aoa} and \ac{aod} perturb the spatial covariance, whose derivatives generate the angular information matrix
$\bm\Gamma_{\text d,k}$. The vector $\bm b_{\text d,k}$ quantifies how 
angular variations couple with the unknown channel magnitude. Moreover, the
rank-one block $\zeta_k\bm\alpha_k\bm\alpha_k^{\text T}$ shows that the
marginalized data identify only $|\alpha_k|$, not its phase, so the data FIM
must be combined with the pilot FIM before eliminating the shared amplitude.
This is consistent with~\cite{SonYuXu:J26}, where only the $|\alpha_k|$ associated with Gaussian information is
estimated, while its phase is not identified from the marginalized data.
\end{remark}

Theorem~\ref{thm:mimo_s1_pilot} and Lemma~\ref{thm:mimo_s1_data}
characterize the pilot and marginalized-data information, respectively.
Because the two observations share the same nuisance amplitude $\bm\alpha_k$,
their \acp{fim} should be summed before applying the Schur complement.
For compactness, define the per-pilot-symbol intensities
\begin{equation}
    \gamma_{\text p,k}^{i}
    \triangleq
    {T_{\text p}^{-1}}
    [\bm J_{\text p}^{\text e}(\bm\theta_k)]_{i,i},
    \qquad i\in\{\tau,\varphi,\psi\}.
    \label{eq:mimo_pilot_gamma_angle}
\end{equation}

\begin{theorem}
\label{thm:mimo_s1}
Under Scenario~1, the channel-domain \ac{efim} obtained from the joint pilot
and data observations is
\begin{equation}
    \bm J_{1}^{\text e}(\bm \theta_k)
    =
    \text{diag}\{\lambda_k^\tau,\lambda_k^\varphi,\lambda_k^\psi\},
    \label{eq:main_total_schur}
\end{equation}
where the delay, \ac{aoa}, and \ac{aod} intensities are
\begin{equation}
\begin{aligned}
\lambda_k^\tau
&=
T_{\text p}\,\gamma_{\text{p},k}^{\tau},\\
\lambda_k^\varphi
&=
T_{\text p}\,\gamma_{\text{p},k}^{\varphi}
+
T_{\text d}\,\gamma_k^{\varphi}(\bm R_{\text d}),\\
\lambda_k^\psi
&=
T_{\text p}\,\gamma_{\text{p},k}^{\psi}
+
T_{\text d}\,\gamma_k^{\psi}(\rho,\bm R_{\text d}),
\end{aligned}
\label{eq:mimo_lambda_s1}
\end{equation}
where $T_{\text p}=\rho T$, $T_{\text d}=(1-\rho)T$,
$\gamma_{\text p,k}^{i}$ is defined in \eqref{eq:mimo_pilot_gamma_angle}, and
\begin{equation}
    \gamma_k^{\varphi}(\bm R_{\text d})
    =[\bm\Gamma_{\text d,k}]_{\varphi,\varphi}.
    \label{eq:mimo_gamma_data_aoa}
\end{equation}
The effective \ac{aod} angular intensity is
\begin{equation}
\gamma_k^{\psi}(\rho,\bm R_{\text d})
=[\bm\Gamma_{\text d,k}]_{\psi,\psi}
-\frac{T_{\text d}[\bm b_{\text d,k}]_{\psi}^2\|\bm\alpha_k\|^2}
       {j_{\text p,k}^{\alpha}+T_{\text d}\zeta_k\|\bm\alpha_k\|^2},
\label{eq:mimo_gamma_data}
\end{equation}
where $j_{\text p,k}^{\alpha}=2T_{\text p}NM_{\text r}/\sigma_v^2$ is the
isotropic pilot amplitude information. The constant-modulus receive steering
vector gives $[\bm b_{\text d,k}]_{\varphi}=0$ and
$[\bm\Gamma_{\text d,k}]_{\varphi,\psi}=0$, as shown in
Appendix~\ref{app:proof_joint_efim}, whereas
$[\bm b_{\text d,k}]_{\psi}$ is generally nonzero when the transmit beampattern
$\bm a_{\text t}^{\text H}(\psi)\bm R_{\text d}\bm a_{\text t}(\psi)$ varies
with $\psi$.
\end{theorem}

\emph{Proof:} See Appendix~\ref{app:proof_joint_efim}.

\begin{remark}
Theorem~\ref{thm:mimo_s1} shows why the Scenario~1 gain is asymmetric.
Marginalizing Gaussian data removes the subcarrier phase rotations, so the data
\ac{fim} has a zero delay block and ranging remains pilot-only. The received
data covariance still depends on the spatial channel, and hence it contributes
angular information in \ac{mimo}. This covariance identifies the channel
magnitude but not its phase, which is why the pilot and data \acp{fim} must be
summed before the common amplitude is eliminated. After this Schur complement,
the \ac{aoa} term is directly additive because it is decoupled from the
amplitude, whereas the \ac{aod} term suffers the penalty in
\eqref{eq:mimo_gamma_data} when the transmit beampattern varies with $\psi$.
The penalty is reduced by pilot amplitude information and vanishes for an
angle-independent transmit beampattern. In \ac{siso}, the marginalized data
covariance is scalar and position-independent, so Scenario~1 reduces to the
pilot-only delay bound.
\end{remark}

The channel-domain \ac{efim} in Theorem~\ref{thm:mimo_s1} is expressed in
$\bm \theta_k = [\tau_k,\varphi_k,\psi]^\text{T}$. Applying the geometric Jacobian of
Section~\ref{subsec:geo_params} by the chain rule maps it to the position domain.
Each channel parameter has a geometry-fixed gradient direction.

\begin{corollary}
\label{thm:mimo_s1_position}
Under Scenario~1, the \ac{efim} of target localization by joint observation of pilot and data is given by
\begin{equation}
\bm J_{1}^{\text e}(\bm p)
=
\sum_{k=1}^{K}\left(
\mu_k^{\tau}\bar{\bm q}_k\bar{\bm q}_k^\text{T}
+ \mu_k^{\varphi}\bm e_{\text{r},k}^{\perp}{\bm e_{\text{r},k}^{\perp}}^\text{T}
\right)
+ \mu^{\psi}\bm e_{\text{t}}^{\perp}{\bm e_{\text{t}}^{\perp}}^\text{T},
\label{eq:mimo_s1_efim}
\end{equation}
where $\mu^\psi=\sum_k\mu_k^\psi$ and
\begin{equation}
    \mu_k^\tau=\frac{4}{c^2}\cos^2\!\left(\frac{\psi-\varphi_k}{2}\right)\lambda_k^\tau,
    ~~
    \mu_k^\varphi=\frac{\lambda_k^\varphi}{d_{\text r,k}^2},
    ~~
    \mu_k^\psi=\frac{\lambda_k^\psi}{d_{\text t}^2}.
    \label{eq:mimo_position_intensities}
\end{equation}
\end{corollary}
It remains to convert a $2\times2$ position \ac{efim} into the localization
bound. Since the \ac{efim} is a sum of rank-one directional terms, its inverse
trace admits a compact closed form.

\begin{proposition}
\label{prop:directional_speb}
Let a $2\times2$ position \ac{efim} be a sum of rank-one terms with directional
index set $\{(c_m,\omega_m)\}_{m=1}^{M}$, where $c_m$ and $\omega_m$ are the
intensity and direction angle of the $m$-th term, i.e., $\bm J=\sum_m c_m\bm e(\omega_m)\bm e^{\text T}(\omega_m)$. The corresponding \ac{speb} is given by
\begin{equation}
    \mathcal{P}(\bm p)
    =
    \frac{\sum_{m=1}^{M}c_m}
    {\sum_{1\le m<n\le M}c_m c_n\sin^2(\omega_m-\omega_n)},
    \label{eq:speb_rankone}
\end{equation}
provided the denominator is nonzero.
\end{proposition}

\emph{Proof:}
For a $2\times2$ \ac{efim} for target localization,
$\text{tr}(\bm J)=\sum_m c_m$ and
$\det(\bm J)=\sum_{m<n}c_mc_n\sin^2(\omega_m-\omega_n)$. Since
$\text{tr}(\bm J^{-1})=\text{tr}(\bm J)/\det(\bm J)$,
\eqref{eq:speb_rankone} follows\cite{SheWin:J10a}. \hfill$\square$

\begin{remark}
The numerator of \eqref{eq:speb_rankone} represents the total information,
whereas the denominator characterizes its directional spread. For two
directions, the \ac{speb} is minimized when they are orthogonal and diverges
when they become collinear. Consequently, a single \ac{mimo} receiver can
localize the target using two non-collinear directions among the bistatic
delay, \ac{aoa}, and \ac{aod} directions. In contrast, a \ac{siso} receiver
provides only the bistatic-delay direction and therefore requires at least two
spatially separated receivers.
\end{remark}

Applying \eqref{eq:speb_rankone} to
\eqref{eq:mimo_s1_efim}, the directional information set in Scenario~1 is given by
\begin{equation}
\mathcal I_{1}
=
\{(\mu_k^\tau,\bar\phi_k)\}_{k=1}^K
\cup
\{(\mu_k^\varphi,\varphi_k+\tfrac{\pi}{2})\}_{k=1}^K
\cup
\{(\mu^\psi,\psi+\tfrac{\pi}{2})\}.
\label{eq:mimo_s1_direction_set}
\end{equation}
Then, the \ac{speb} for target localization in Scenario~1 is
\begin{equation}
    \mathcal{P}_{1}(\bm p)
    =
    \frac{\sum_{k=1}^{K}\big(\mu_k^\tau+\mu_k^\varphi\big)+\mu^\psi}
    {\sum_{1\le m<n\le 2K+1}
    c_m c_n\sin^2(\omega_m-\omega_n)},
    \label{eq:speb_s1}
\end{equation}
where
$\{(c_m,\omega_m)\}_{m=1}^{2K+1}=\mathcal I_{1}$ and 
\eqref{eq:speb_s1} shows that the pilot fraction $\rho$ shapes not only the
total localization information but also its distribution across directions.
Larger $\rho$ strengthens the pilot-driven delay direction, while smaller $\rho$
allocates more observations to the data-driven angular directions, so the
optimal allocation depends on both the intensities and the directional geometry.

\subsubsection{Scenario 2: Joint Data-Aided Sensing and Decoding}
Scenario~2 considers the reliable-recovery regime in which the transmitted data
realization is decoded without error and then treated as known, which provides
an information-upper-bound benchmark for data-aided sensing with decoded
data~\cite{ZhaKamAlo:J25}. Specifically,
let $\bm S_{\text d}$ denote a realization of the random data matrix
$\bm{\mathsf S}_{\text d}$. Under correct decoding,
$\hat{\bm{\mathsf S}}_{\text d}=\bm S_{\text d}$, so the realized full waveform
$\bm S_2=[\bm S_{\text p},\bm S_{\text d}]$ is known and contributes through
the same mean-information mechanism as the pilots. Conditioned on
$\bm S_{\text d}$, its EFIM follows from the pilot derivation by replacing
$\bm S_{\text p}$ with $\bm S_2$. For offline covariance design, before the
realization is available, define the random full-frame waveform
$\bm{\mathsf S}_2=[\bm S_{\text p},\bm{\mathsf S}_{\text d}]$ and average its
FIM over $\bm{\mathsf S}_{\text d}$ using
\begin{equation}
\begin{aligned}
    &\mathbb E_{\bm{\mathsf S}_{\text d}}
    \{\bm{\mathsf S}_2\bm{\mathsf S}_2^{\text H}\}
    =T(\bm I_N\otimes\bm R_2),\\
    &\bm R_2 = \frac{\rho}{M_{\text t}}\bm I_{M_{\text t}}
      +(1-\rho)\bm R_{\text d}.
\end{aligned}
    \label{eq:mimo_s2_effective_covariance}
\end{equation}
The expectation above gives the corresponding offline covariance-design
benchmark.

\begin{theorem}
\label{thm:mimo_s2}
Under Scenario~2, the \ac{efim} of geometric channel parameter $\bm \theta_k$ is given by
\begin{equation}
    \bar{\bm J}_{2}^{\text e}(\bm \theta_k)
    =
    \text{diag}\{\tilde\lambda_k^\tau,\tilde\lambda_k^\varphi,\tilde\lambda_k^\psi\}.
    \label{eq:mimo_s2_expected_diag}
\end{equation}
Define the three full-frame transmit moments as
\begin{equation}
\begin{aligned}
    q_0
    &=\rho+(1-\rho)
      \bm a_{\text t}^{\text H}\bm R_{\text d}\bm a_{\text t},\\
    q_1
    &=\frac{j\rho\nu_{\text t}s_1(M_{\text t})}{M_{\text t}}
      +(1-\rho)\bm a_{\text t}^{\text H}
      \bm R_{\text d}\dot{\bm a}_{\text t},\\
    q_2
    &=\frac{\rho\nu_{\text t}^2s_2(M_{\text t})}{M_{\text t}}
      +(1-\rho)\dot{\bm a}_{\text t}^{\text H}
      \bm R_{\text d}\dot{\bm a}_{\text t},
\end{aligned}
\label{eq:mimo_s2_moments}
\end{equation}
where $\bm a_{\text t}=\bm a_{\text t}(\psi)$ and
$\dot{\bm a}_{\text t}=\partial\bm a_{\text t}(\psi)/\partial\psi$, the three diagonal intensities are given by
\begin{equation}
\begin{aligned}
\tilde\lambda_k^\tau
&=\frac{2T|\alpha_k|^2}{\sigma_v^2}M_{\text r}q_0\omega^2
\Big[s_2(N)-\frac{s_1^2(N)}{N}\Big],\\
\tilde\lambda_k^\varphi
&=\frac{2T|\alpha_k|^2}{\sigma_v^2}Nq_0\nu_{\text r}^2
\Big[s_2(M_{\text r})-\frac{s_1^2(M_{\text r})}{M_{\text r}}\Big],\\
\tilde\lambda_k^\psi
&=\frac{2T|\alpha_k|^2}{\sigma_v^2}NM_{\text r}
\Big[q_2-\frac{|q_1|^2}{q_0}\Big].
\end{aligned}
    \label{eq:mimo_s2_lambda}
\end{equation}
\end{theorem}

\emph{Proof:} See Appendix~\ref{app:proof_s2_efim}.

\begin{remark}
Theorem~\ref{thm:mimo_s2} has the same centered-aperture structure as
Theorem~\ref{thm:mimo_s1_pilot}, but the isotropic pilot covariance is replaced
by the effective full-frame covariance $\bm R_2$. Thus, unlike the marginalized
data of Scenario~1, the decoded data preserve the subcarrier phases and
contribute mean information to the delay as well as both angles. The effective
beampattern gain $q_0$ scales the delay and \ac{aoa} information, whereas the
Schur-complement term in $\tilde\lambda_k^\psi$ is the transmit aperture
measured after removing the component indistinguishable from the unknown
amplitude. Setting $\bm R_2=\bm R_{\text p}$ recovers
Theorem~\ref{thm:mimo_s1_pilot} with $T_{\text p}$ replaced by $T$.
\end{remark}

For the expected known-data benchmark, no new geometry is needed. The position
transformation of Corollary~\ref{thm:mimo_s1_position} and the rank-one inversion
of Proposition~\ref{prop:directional_speb} apply unchanged, with the Scenario~1
intensities $\lambda_k^i$ replaced by their known-data counterparts
$\tilde\lambda_k^i$. The resulting Scenario~2 position \ac{efim} is
\begin{equation}
\bm J_{2}^{\text e}(\bm p)
=
\sum_{k=1}^{K}\left(
\tilde\mu_k^{\tau}\bar{\bm q}_k\bar{\bm q}_k^\text{T}
+ \tilde\mu_k^{\varphi}\bm e_{\text{r},k}^{\perp}{\bm e_{\text{r},k}^{\perp}}^\text{T}
\right)
+ \tilde\mu^{\psi}\bm e_{\text{t}}^{\perp}{\bm e_{\text{t}}^{\perp}}^\text{T},
\label{eq:mimo_s2_efim}
\end{equation}
with $\tilde\mu^\psi=\sum_k\tilde\mu_k^\psi$ and the per-link intensities
$\tilde\mu_k^i$ obtained from~\eqref{eq:mimo_position_intensities} under the
substitution from $\lambda_k^i$ to $\tilde\lambda_k^i$. Listing the
corresponding directions in the Scenario~2 index set
\begin{equation}
\mathcal I_{2}
=
\{(\tilde\mu_k^\tau,\bar\phi_k)\}_{k=1}^K
\cup
\{(\tilde\mu_k^\varphi,\varphi_k+\tfrac{\pi}{2})\}_{k=1}^K
\cup
\{(\tilde\mu^\psi,\psi+\tfrac{\pi}{2})\}.
\label{eq:mimo_s2_direction_set}
\end{equation}
Substitution of $\mathcal I_{2}$ into
Proposition~\ref{prop:directional_speb} gives the \ac{speb}
$\mathcal{P}_{2}(\bm p)$, which differs from $\mathcal{P}_{1}(\bm p)$ only in the
intensities $\tilde\mu_k^i$. Because the decoded data add mean-based delay and
angular information on top of Scenario~1, the known-data benchmark satisfies
$\bm J_{2}^{\text e}(\bm p)\succeq\bm J_{1}^{\text e}(\bm p)$ and
$\mathcal{P}_{2}(\bm p)\le\mathcal{P}_{1}(\bm p)$.

\subsubsection{SISO Special Case}
Setting $M_{\text t}=M_{\text r}=1$ removes the angular parameters, so each
link provides only bistatic-delay information. In Scenario~1, the pilot and
data FIMs are still combined, but
$\tilde Q_k=|\alpha_k|^2+\sigma_v^2$ is position independent and its radial
amplitude information does not alter the pilot delay EFIM. In Scenario~2,
$R_{\text p}=R_{\text d}=1$, so the entire correctly decoded frame
contributes mean-based delay information.

\begin{corollary}
\label{cor:siso_special}
Under \ac{siso} model case, the \acp{efim} for target localization in both scenarios reduce to
\begin{equation}
    ~~\bm J_1^{\text e}(\bm p)
    =\sum_{k=1}^{K}\mu_k^\tau
      \bar{\bm q}_k\bar{\bm q}_k^{\text T},~~~~
    \bm J_2^{\text e}(\bm p)
    =\sum_{k=1}^{K}\tilde\mu_k^\tau
      \bar{\bm q}_k\bar{\bm q}_k^{\text T},
    \label{eq:siso_position_efims}
\end{equation}
where
\begin{equation}
\begin{aligned}
   & \mu_k^\tau=\rho\tilde\mu_k^\tau,~~~~~
    \tilde\mu_k^\tau
    =\frac{4}{c^2}
      \cos^2\Big(\frac{\psi-\varphi_k}{2}\Big)
      T\gamma_{\text p,k}^\tau,\\
    &~~~~~~~\gamma_{\text p,k}^\tau
    =\frac{|\alpha_k|^2(2\pi\Delta f)^2N(N^2-1)}
      {6\sigma_v^2}.
\end{aligned}
    \label{eq:siso_position_intensities}
\end{equation}
\end{corollary}
Applying Proposition~\ref{prop:directional_speb} to the $K$ bistatic-delay
directions gives
\begin{equation}
    \mathcal P_{1}(\bm p)
    =
    \frac{\sum_{k=1}^{K}\mu_k^\tau}
    {\sum_{1\le i<j\le K}
    \mu_i^\tau\mu_j^\tau
    \sin^2(\bar\phi_i-\bar\phi_j)}.
    \label{eq:siso_s1_speb}
\end{equation}
and, since $\tilde\mu_k^\tau=\mu_k^\tau/\rho$,
 $\mathcal P_{2}(\bm p;\rho)
    =\rho\,\mathcal P_{1}(\bm p;\rho)
    =\mathcal P_{1}(\bm p;1)$,
provided the delay directions span the 2D position space.

\subsection{Achievable Communication Rate}\label{sec:comm}
The sensing and communication functions operate over the same target-reflected bistatic
channel. On subcarrier $n$, the reflected channel is
$\bm H_k[n]\triangleq e^{-j2\pi n\Delta f\tau_k}\bm H_k
=\alpha_k e^{-j2\pi n\Delta f\tau_k}\bm A(\bm\vartheta_k)$, and
$\hat{\bm H}_k[n]$ denotes its pilot-based estimate. Each receiver
uses the remaining $T_{\text d}=(1-\rho)T$ symbols to decode the data conveyed.

\begin{proposition}
\label{prop:rates}
Treating the residual channel-estimation error as worst-case Gaussian noise\cite{HasHoc:J13},
define the estimated effective SNR on subcarrier $n$ as
\begin{equation}
    \Gamma_{\text c,k}[n]
    \triangleq
    \kappa(\rho)
    \text{tr}\big(
    \hat{\bm H}_k[n]\bm R_{\text d}
    \hat{\bm H}_k^{\text H}[n]\big).
    \label{eq:mimo_effective_snr}
\end{equation}
The achievable ergodic rate of receiver~$k$ is
\begin{equation}
    \bar R_k(\rho,\bm R_{\text d})
    =
    \frac{1-\rho}{N}
    \sum_{n\in\mathcal N}
    \mathbb E\left\{
    \log_2\big(1+\Gamma_{\text c,k}[n]\big)
    \right\}.
    \label{eq:mimo_rate_general}
\end{equation}
The expectation is over $\hat{\bm H}_k[n]$, and the effective
inverse-noise factor under the orthogonal pilots
$\bm R_{\text p}=\bm I_{M_{\text t}}/M_{\text t}$ is
\begin{equation}
    \kappa(\rho)
    \triangleq
    \Big[
    \sigma_v^2\big(
    1+T_{\text p}^{-1}
    \text{tr}(\bm R_{\text d}\bm R_{\text p}^{-1})
    \big)
    \Big]^{-1}
    =
    \frac{\rho T}{\sigma_v^2(\rho T+M_{\text t})}.
    \label{eq:effective_comm_snr}
\end{equation}
For \ac{siso}, $M_{\text t}=M_{\text r}=1$ and $R_{\text d}=1$, yielding
$\Gamma_{\text c,k}[n]=\kappa(\rho)|\hat h_k[n]|^2$, which specializes
\eqref{eq:mimo_rate_general} to the scalar  channel.
\end{proposition}

\emph{Proof:}
Orthogonal-pilot \ac{ml} estimation of $\bm H_k[n]$ gives the 
channel-error covariance
$\sigma_v^2\bm R_{\text p}^{-1}/T_{\text p}$. Hence the residual error plus
 noise has variance
$\sigma_v^2[1+T_{\text p}^{-1}
\text{tr}(\bm R_{\text d}\bm R_{\text p}^{-1})]$.
Substitution into the Gaussian-input rate and multiplication by 
 $1-\rho$, followed by the matrix determinant lemma for the
rank-one reflected channel, gives \eqref{eq:mimo_rate_general}. 
\hfill$\square$

The communication rate and localization accuracy are coupled through the same
target-reflected channel. Increasing $\rho$ improves its estimate but shortens
the data phase, while $\bm R_{\text d}$ simultaneously shapes the reflected
communication beampattern and the sensing intensities. The target-dependent
gain and angles therefore affect both the communication constraint and the
localization bound. In
\ac{siso}, the covariance-design degree of freedom disappears and only the
pilot fraction remains.

\subsection{Extension to Finite-Alphabet Signaling}\label{sec:ext}
The preceding performance limits are derived under Gaussian signaling, which
enables closed-form covariance \acp{fim}. For finite-alphabet signaling, with
symbols drawn independently and uniformly from $\mathcal A$, the data likelihood
instead becomes a finite Gaussian mixture. Define the vectorized data-phase
observation and its conditional mean as
$\tilde{\bm{\mathsf y}}_{\text d,k}\triangleq
\text{vec}(\tilde{\bm{\mathsf Y}}_{\text d,k})
\in\mathbb C^{NM_{\text r}T_{\text d}}$ and
$\bm m_k(\bm S_{\text d})\triangleq
\text{vec}\left(\alpha_k(\bm D(\tau_k)\otimes
\bm A(\bm\vartheta_k))\bm S_{\text d}\right)$, respectively. Then,
marginalizing over all possible transmitted data matrices yields the
Gaussian-mixture likelihood
\begin{equation}\label{eq:finite_marginal}
    f(\tilde{\bm{\mathsf y}}_{\text d,k}; \bm p, \bm\kappa)\propto
    \sum_{\bm S_{\text d} \in \mathcal{A}^{NM_{\text t} T_{\text d}}}
    \mathcal{CN}\big(
    \tilde{\bm{\mathsf y}}_{\text d,k};
    \bm m_k(\bm S_{\text d}),
    \sigma_v^2\bm I_{NM_{\text r}T_{\text d}}
    \big).
\end{equation}
The corresponding finite-alphabet data rate is
$\bar R_k^{\text{FA}}(\rho)=(1-\rho)I(\bm{\mathsf S}_{\text d};
\tilde{\bm{\mathsf y}}_{\text d,k}|\hat{\bm H}_{k})$, which saturates at the
constellation entropy, where $\hat{\bm H}_{k}$ denotes the channel estimated
from the pilot phase. The constrained formulation
in~\eqref{eq:constrained_bound_problem} then applies after replacing
$\bar R(\bm x)$ and $\mathcal{P}(\bm p;\bm x)$ by their finite-alphabet
counterparts. However, the closed-form Gaussian covariance \acp{fim} above do
not carry over directly. For localization-bound derivation, Scenario~1 requires
the \ac{fim} of the mixture likelihood in~\eqref{eq:finite_marginal}, whereas
Scenario~2 uses the conditional known-waveform \ac{fim} averaged over the
finite-alphabet data symbols. For estimator design,
\eqref{eq:finite_marginal} can be handled by \ac{em} or message-passing
approximations in Scenario~1, while Scenario~2 uses the joint target
localization and decoding algorithm with a symbol-wise \ac{map} or soft decoder.

\begin{remark}
Finite-alphabet signaling introduces an information--complexity tension.
Unlike Gaussian signaling, which is invariant to arbitrary phase rotations, a
finite constellation generally possesses only discrete rotational symmetries.
Hence, the marginalized likelihood may retain the delay-induced phase rotations,
allowing unknown data to provide direct delay and position information even in
\ac{siso} systems. However, the Gaussian mixture contains
$|\mathcal A|^{NM_{\text t}T_{\text d}}$ components, so exact marginalization
and the \ac{fim} evaluation scale exponentially with the signal dimension and
become computationally prohibitive for large $N$, $M_{\text t}$, or
$T_{\text d}$. Therefore, exploiting additional information generally requires
approximate inference, while discrete constellation symmetries may still cause
phase ambiguities.
\end{remark}

\section{Communication-Constrained Target Localization Optimization Problem}
\label{sec:optimization}
This section optimizes the performance limits of target localization under communication-rate
constraints by jointly designing the pilot fraction and  transmitted data covariance.

\subsection{Unified Constrained Optimization}
Let $\mathcal P_\ell(\bm p;\bm x)$ denote the \ac{speb} in
Scenario~$\ell\in\{1,2\}$, where
$\bm x=(\rho,\bm R_{\text d})$ for \ac{mimo} and $\bm x=\rho$ for \ac{siso}.
For a required data rate $R_{\text{th}}$, the communication-constrained target 
localization design is obtained from
\begin{equation}\label{eq:unified_constrained_problem}
\begin{aligned}
\mathscr P_{\ell}(R_{\text{th}}):\quad
\min_{\bm x\in \mathcal{X}}\quad
&\mathcal P_{\ell}(\bm p;\bm x)\\
\text{s.t.}\quad
&\min_{k=1,\ldots,K}\bar R_k(\bm x)\ge R_{\text{th}}
\end{aligned}
\end{equation}
Equivalently, the broadcast-rate constraint requires
$\bar R_k(\bm x)\ge R_{\text{th}}$ for every receiver $k$.
The orthogonality condition
in~\eqref{eq:orthogonal_pilot_main} requires $T_{\text p}\ge M_{\text t}$.
Thus, assuming $T\ge M_{\text t}$,  then we have $\rho_{\min}={M_{\text t}}/{T}$, $\mathcal X_{\text S}=[\rho_{\min},1]$,  and
$\mathcal X_{\text M}
=[\rho_{\min},1]\times
\{\bm R_{\text d}\succeq\bm0:
\text{tr}(\bm R_{\text d})=1\}$.
In~\eqref{eq:unified_constrained_problem}, $\mathcal X=\mathcal X_{\text M}$ for \ac{mimo} and
$\mathcal X=\mathcal X_{\text S}$ for \ac{siso}.
Varying $R_{\text{th}}$ over the feasible broadcast data rates traces the
joint \ac{sac} bound under communication constraints.
For a fixed $\bm R_{\text d}$, define the broadcast-feasible pilot set
$\mathcal F(\bm R_{\text d})=\left\{\rho\in[\rho_{\min},1]:
\bar R_k(\rho,\bm R_{\text d})\ge R_{\text{th}},~\forall k\right\}$.
The time-allocation update
is the one-dimensional constrained search
\begin{equation}
    \rho^{(t+1)}
    =\arg\min_{\rho\in\mathcal F(\bm R_{\text d}^{(t)})}
    \mathcal P_\ell(\bm p;\rho,\bm R_{\text d}^{(t)}).
    \label{eq:rho_exact_update}
\end{equation}
In implementation, the feasible connected components are identified by a
one-dimensional grid, and the best component is refined by standard
one-dimensional search methods, such as bisection or golden-section search
\cite{BoyVan:B04}.

\subsection{Alternating Optimization for MIMO Case}
\label{subsec:mimo_bcd}
The \ac{efim} for target localization in the \ac{mimo} case has the rank-one form in
\eqref{eq:unified_efim}. For fixed geometry, $\rho$ and $\bm R_{\text d}$
change only its directional intensities. These are enumerated by
$\mathcal I_1$ in~\eqref{eq:mimo_s1_direction_set} and $\mathcal I_2$
in~\eqref{eq:mimo_s2_direction_set} for Scenarios~1 and~2, respectively.
We therefore adopt a two-block \ac{sca} method, alternating
between the scalar pilot fraction and the transmit covariance
\cite{RazHonLuo:J13,ScuFacLamSon:J17}.
Consider
$\bm J_{\ell}^{\text e}(\bm p)=\sum_m c_m(\rho,\bm R_{\text d})
\bm v_m\bm v_m^{\text T}$, where $\|\bm v_m\|=1$ and the directions
$\bm v_m$ are fixed by the geometry. Taking the differential \ac{wrt}
the design variables $(\rho,\bm R_{\text d})$, since
$\mathcal P_\ell(\bm p)=\text{tr}\{(\bm J_\ell^{\text e}(\bm p))^{-1}\}$ and
$\text{d}\bm J_\ell^{\text e}
=\sum_m\bm v_m\bm v_m^{\text T}\text{d}c_m$, its differential is
\begin{equation}
\text{d}\mathcal P_\ell
=-\text{tr}\left\{
(\bm J_\ell^{\text e})^{-1}\text{d}\bm J_\ell^{\text e}
(\bm J_\ell^{\text e})^{-1}\right\}
=-\sum_m
\bm v_m^{\text T}(\bm J_\ell^{\text e})^{-2}\bm v_m\,\text{d}c_m.
\label{eq:speb_differential}
\end{equation}
Accordingly, the \ac{speb} sensitivity to the $m$-th intensity is
\begin{equation}
    s_m
    \triangleq
    -\frac{\partial\mathcal P_\ell}{\partial c_m}
    =
    \bm v_m^{\text T}(\bm J_{\ell}^{\text e}(\bm p))^{-2}\bm v_m .
    \label{eq:speb_sensitivity}
\end{equation}
Thus a larger $s_m$ identifies a direction in which additional information
produces a larger local reduction in \ac{speb}. After the $\rho$-update, let
$\bm c^{(t)}\triangleq
[c_m(\rho^{(t+1)},\bm R_{\text d}^{(t)})]_m$, and evaluate $s_m^{(t)}$ at the
same point. The first-order expansion is
\begin{equation}
\mathcal P_\ell(\bm c)
\approx \mathcal P_\ell(\bm c^{(t)})
-\sum_m s_m^{(t)}(c_m-c_m^{(t)}).
\label{eq:speb_first_order}
\end{equation}
Minimizing this
local model is equivalent to maximizing the sensitivity-weighted information
\cite{SunBabPal:J17}
\begin{equation}
    \mathcal W_\ell^{(t)}(\bm R_{\text d})
    \triangleq
    \sum_m s_m^{(t)}c_m(\rho^{(t+1)},\bm R_{\text d}).
    \label{eq:weighted_intensity_objective}
\end{equation}
The chain rule gives
$-\nabla_{\bm R_{\text d}}\mathcal P_\ell=\sum_m s_m\nabla_{\bm R_{\text d}}c_m
=\nabla_{\bm R_{\text d}}\mathcal W_\ell$ at the current iterate. Hence
the gradient of $\mathcal W_\ell^{(t)}$ is the exact first-order descent
direction of the \ac{speb}.

For fixed $\rho$, introduce the rate-feasible covariance set
$\mathcal C(\rho)=\left\{\bm R_{\text d}\succeq\bm0:
\text{tr}(\bm R_{\text d})=1,~
\bar R_k(\rho,\bm R_{\text d})\ge R_{\text{th}},~\forall k\right\}$.
The rate in~\eqref{eq:mimo_rate_general} is concave in $\bm R_{\text d}$ for fixed
$\rho$, so $\mathcal C(\rho)$ is convex. A  proximal
minorization step is obtained from the standard \ac{sca} construction
\cite{SunBabPal:J17,ScuFacLamSon:J17}
\begin{equation}
\begin{aligned}
\hat{\mathcal W}_\ell^{(t)}(\bm R_{\text d})
&=\mathcal W_\ell^{(t)}(\bm R_{\text d}^{(t)})
+\left\langle\nabla_{\bm R_{\text d}}\mathcal W_\ell^{(t)}
(\bm R_{\text d}^{(t)}),
\bm R_{\text d}-\bm R_{\text d}^{(t)}\right\rangle\\
&\quad-\frac{L_t}{2}\|\bm R_{\text d}
-\bm R_{\text d}^{(t)}\|_{\text F}^2,\\
\breve{\bm R}_{\text d}^{(t+1)}
&=\arg\max_{\bm R_{\text d}\in\mathcal C(\rho^{(t+1)})}
\hat{\mathcal W}_\ell^{(t)}(\bm R_{\text d}).
\end{aligned}
\label{eq:proximal_covariance_update}
\end{equation}
Starting from a positive $L_t$, backtracking increases it until the candidate
is feasible and satisfies
$\mathcal P_\ell(\bm p;\rho,\breve{\bm R}_{\text d}^{(t+1)})
\le\mathcal P_\ell(\bm p;\rho,\bm R_{\text d}^{(t)})
-\delta\|\breve{\bm R}_{\text d}^{(t+1)}
-\bm R_{\text d}^{(t)}\|_{\text F}^2$ for some $\delta>0$. The strongly
concave subproblem in~\eqref{eq:proximal_covariance_update} then has a unique
solution~\cite{RazHonLuo:J13,ScuFacLamSon:J17}. The initial point is found by maximizing
$\min_k\bar R_k(\rho,\bm R_{\text d})$ over $\mathcal X_{\text M}$, using the
same one-dimensional search in $\rho$ and a convex covariance subproblem. If
the resulting rate is below $R_{\text{th}}$, Problem~\eqref{eq:unified_constrained_problem} is
infeasible.

\begin{algorithm}[t]
\caption{Alternating Proximal-SCA for Rate-Constrained \ac{speb} Minimization}\label{alg:constrained_bound_opt}
\begin{algorithmic}[1]
\STATE \textbf{Input:} $R_{\text{th}}$, $\ell\in\{1,2\}$, tolerance $\epsilon$.
\STATE Find a feasible $(\rho^{(0)},\bm R_{\text d}^{(0)})$ and set $t=0$.
\REPEAT
    \STATE Update $\rho^{(t+1)}$ by~\eqref{eq:rho_exact_update};
    \STATE Evaluate $\bm J_\ell^{\text e}$ and $\{s_m^{(t)}\}$ at
    $(\rho^{(t+1)},\bm R_{\text d}^{(t)})$ using~\eqref{eq:speb_sensitivity};
    \STATE Form $\mathcal W_\ell^{(t)}$ in~\eqref{eq:weighted_intensity_objective};
    \STATE Backtrack $L_t$ and solve~\eqref{eq:proximal_covariance_update};
    \STATE Set $\bm R_{\text d}^{(t+1)}=\breve{\bm R}_{\text d}^{(t+1)}$ and $t = t+1$;
\UNTIL{the relative \ac{speb} reduction and block-step norm are below $\epsilon$}.
\STATE \textbf{Output:} $(\rho^\star,\bm R_{\text d}^\star)$.
\end{algorithmic}
\end{algorithm}

The accepted block updates monotonically decrease the exact
\ac{speb}. Specifically, the pilot-fraction update in
\eqref{eq:rho_exact_update} and the covariance update satisfy
$\mathcal P_\ell(\bm p;\rho^{(t+1)},\bm R_{\text d}^{(t+1)})
\le
\mathcal P_\ell(\bm p;\rho^{(t+1)},\bm R_{\text d}^{(t)})
\le
\mathcal P_\ell(\bm p;\rho^{(t)},\bm R_{\text d}^{(t)})$.
Since the \ac{speb} is nonnegative, the objective-value sequence converges.
This statement concerns only the monotonic convergence of the objective, and convergence of
the design variables is examined numerically in Section~\ref{sec:results}.
General stationary-point guarantees for block \ac{sca} methods
require additional regularity and exact block-solution conditions, as in
standard block-coordinate and \ac{sca} analyses
\cite{RazHonLuo:J13,ScuFacLamSon:J17}.

\subsection{Time Allocation for SISO Case}
\label{subsec:siso_closed_form}
For \ac{siso}, $R_{\text d}=1$ and only $\rho$ remains. Define
$\eta_k\triangleq(4T/c^2)\cos^2((\psi-\varphi_k)/2)\,\gamma_{\text{p},k}^{\tau}$
and
\begin{equation}
    C_{\text S}
    =
    \frac{\sum_{k=1}^{K}\eta_k}
    {\sum_{1\le i<j\le K}
    \eta_i\eta_j\sin^2(\bar\phi_i-\bar\phi_j)} .
    \label{eq:siso_geometry_constant}
\end{equation}
Provided the delay directions span the plane, Corollary~\ref{cor:siso_special}
gives
    $\mathcal P_1(\bm p;\rho)={C_{\text S}}/{\rho}$,
  and
    $\mathcal P_2(\bm p;\rho)=C_{\text S}$.
Let
$\mathcal F_{\text S}=\{\rho\in[N/T,1]:
\bar R_k(\rho)\ge R_{\text{th}},~\forall k\}$. If this set is nonempty, the
Scenario~1 solution is
$    \rho_{\text{S1}}^\star=\max\mathcal F_{\text S},~
    \mathcal P_1^\star={C_{\text S}}/{\rho_{\text{S1}}^\star}.$
Every $\rho\in\mathcal F_{\text S}$ attains the Scenario~2 optimum
$\mathcal P_2^\star=C_{\text S}$. A particular point may therefore be selected
to maximize the broadcast data rate or satisfy a secondary pilot-overhead criterion.

\section{Data-Aided Target Localization Algorithm}\label{sec:DAS}
The data-aided schemes lead to two estimators with distinct uses of data symbols, in which Scenario~1 integrates out the unknown data and exploits their
covariance, whereas Scenario~2 estimates the  transmitted data realization and reuses it as a known counterpart. We assume that the link
observations are available to a fusion processor. Referring to
\eqref{eq:mimo_matrix_model}, define the position-dependent channel matrix
$\bm B_k(\bm p)
    \triangleq
    \bm D(\tau_k(\bm p))\otimes
    \bm A(\bm\vartheta_k(\bm p))$.
The same transmitted data realization $\bm S_{\text d}$ is observed over all
$K$ links, and only $\alpha_k$ and $\bm B_k(\bm p)$ are link dependent.

\begin{algorithm}[t]
\caption{Marginal-Likelihood Based Target Localization Algorithm in Scenario 1}
\label{alg:statistical_localizer}
\begin{algorithmic}[1]
\STATE \textbf{Input:} $\{\tilde{\bm{\mathsf Y}}_{\text p,k},
\tilde{\bm{\mathsf Y}}_{\text d,k}\}_{k=1}^K$, $\bm S_{\text p}$, $\bm R_{\text d}$.
\STATE Form $\{\hat{\bm Q}_{y,k}\}_{k=1}^K$ using~\eqref{eq:sample_data_cov_est};
\STATE Generate a coarse position set from  pilot-only likelihood;
\FOR{each candidate $\bm p$ in the coarse-to-fine search set}
    \FOR{$k=1,\ldots,K$}
        \STATE Compute $\bm X_{\text p,k}(\bm p)$, $z_k(\bm p)$, and $a_k(\bm p)$ as defined above
        \STATE Solve~\eqref{eq:scenario1_magnitude_profile} for $\hat r_k(\bm p)$;
    \ENDFOR
    \STATE Evaluate $\mathcal C_1(\bm p)$ in~\eqref{eq:scenario1_profile_cost};
\ENDFOR
\STATE Refine the best candidate by local minimization of $\mathcal C_1(\bm p)$.
\STATE \textbf{Output:} $\hat{\bm p}_1$ and $\{\hat\alpha_{1,k}\}_{k=1}^K$.
\end{algorithmic}
\end{algorithm}

\subsection{Marginal-Likelihood Based Localization in Scenario 1}
Firstly, marginalizing the Gaussian data symbols gives the zero-mean per-subcarrier covariance matrix with
  $  \bm Q_k(\bm p,r_k)
    \triangleq
    r_k^2\bm A(\bm\vartheta_k(\bm p))\bm R_{\text d}
    \bm A^{\text H}(\bm\vartheta_k(\bm p))
    +\sigma_v^2\bm I_{M_{\text r}}$,
where $r_k=|\alpha_k|$. Since this covariance is common to all data samples,
their sufficient statistic is the sample receive covariance
\begin{equation}
    \hat{\bm Q}_{y,k}
    =
    \frac{1}{NT_{\text d}}
    \sum_{t\in\mathcal T_{\text d}}\sum_{n\in\mathcal N}
    \bm{\mathsf y}_{k}[n,t]\bm{\mathsf y}_{k}^{\text H}[n,t].
    \label{eq:sample_data_cov_est}
\end{equation}
Treating the marginalized data as an unconditional (stochastic) Gaussian
observation with the sample covariance as a sufficient
statistic~\cite{StoNeh:J90}, the negative data log-likelihood is
given by
\begin{equation}
    \Phi_{\text d,k}(\bm p,r_k)
    \triangleq NT_{\text d}\big(\ln\det\bm Q_k(\bm p,r_k)
+\text{tr}(
    \bm Q_k^{-1}(\bm p,r_k)\hat{\bm Q}_{y,k}
    )\big).
    \label{eq:data_marginal_likelihood}
\end{equation}
For a trial position, define
$\bm X_{\text p,k}(\bm p)\triangleq\bm B_k(\bm p)\bm S_{\text p}$,
$z_k(\bm p)\triangleq\text{tr}\!\big(
\bm X_{\text p,k}^{\text H}(\bm p)
\tilde{\bm{\mathsf Y}}_{\text p,k}\big)$, and
$a_k(\bm p)\triangleq\|\bm X_{\text p,k}(\bm p)\|_{\text F}^2$.
Writing $\alpha_k=r_ke^{j\phi_k}$, the pilot residual is minimized for
$\hat\phi_k(\bm p)=\arg z_k(\bm p)$ at any fixed $r_k$. The exact concentrated
magnitude is therefore obtained from the scalar problem
\begin{equation}
\hat r_k(\bm p)
=\arg\min_{r\ge0}
\left\{
\frac{a_k(\bm p)r^2-2r|z_k(\bm p)|}{\sigma_v^2}
+\Phi_{\text d,k}(\bm p,r)
\right\},
\label{eq:scenario1_magnitude_profile}
\end{equation}
and
$\hat\alpha_{1,k}(\bm p)=
\hat r_k(\bm p)e^{j\arg z_k(\bm p)}$.
Then, the resulting position cost and estimate are given by
\begin{equation}
\begin{aligned}
\Phi_{\text p,k}(\bm p)
&\triangleq\frac{1}{\sigma_v^2}
\|\tilde{\bm{\mathsf Y}}_{\text p,k}
-\hat\alpha_{1,k}(\bm p)\bm X_{\text p,k}(\bm p)\|_{\text F}^2,\\
\mathcal C_1(\bm p)
&\triangleq\sum_{k=1}^{K}\left[
\Phi_{\text p,k}(\bm p)
+\Phi_{\text d,k}(\bm p,\hat r_k(\bm p))\right],\\
\hat{\bm p}_1&=\arg\min_{\bm p}\mathcal C_1(\bm p).
\end{aligned}
\label{eq:scenario1_profile_cost}
\end{equation}
The resulting nonconvex localization problem is solved by a
coarse-to-fine search over the target region, with the inner scalar
minimization in~\eqref{eq:scenario1_magnitude_profile} evaluated at each trial
position and the best candidate refined locally, as in standard direct
localization methods~\cite{BosKraYar:J16}.

\begin{algorithm}[t]
\caption{Joint MAP Target Localization and Data-Recovery Algorithm in Scenario 2}
\label{alg:joint_localizer}
\begin{algorithmic}[1]
\STATE \textbf{Input:} $\{\tilde{\bm{\mathsf Y}}_k\}_{k=1}^K$,
$\bm S_{\text p}$, $\mathcal C_{\text d}$, and $f(\bm S_{\text d})$.
\STATE Initialize $(\bm p^{(0)},\{\alpha_k^{(0)}\})$ using Algorithm~\ref{alg:statistical_localizer};
\STATE Set $t=0$;
\REPEAT
    \STATE Recover the common $\hat{\bm S}_{\text d}^{(t)}$ using~\eqref{eq:scenario2_data_update};
    \STATE Form $\hat{\bm S}^{(t)}_2=[\bm S_{\text p},\hat{\bm S}_{\text d}^{(t)}]$;
    \STATE Update $\bm p^{(t+1)}$ using~\eqref{eq:mstep_position};
    \STATE Update $\{\alpha_k^{(t+1)}\}$ using~\eqref{eq:alpha_concentrated_est};
    \STATE Accept an approximate block update only if the joint cost does not increase;
    \STATE Set $t= t+1$;
\UNTIL{the position and recovered data waveform stabilize.}
\STATE \textbf{Output:} $\hat{\bm p}_2$, $\{\hat\alpha_{2,k}\}$, and
$\hat{\bm S}_{\text d}$.
\end{algorithmic}
\end{algorithm}

\subsection{Joint Target Localization and Data Recovery in Scenario 2}
Let $\mathcal C_{\text d}$ denote the admissible set of the broadcast data
matrix. Scenario~2 jointly estimates the position, link amplitudes, and
data waveform. Define its negative joint log-posterior as
\begin{equation}
\mathcal L_2(\bm p,\{\alpha_k\},\bm S_{\text d})
\triangleq\frac{1}{\sigma_v^2}\sum_{k=1}^{K}
\left\|\tilde{\bm{\mathsf Y}}_k
-\alpha_k\bm B_k(\bm p)\bm S_2
\right\|_{\text F}^2
-\ln f(\bm S_{\text d}),
\label{eq:scenario2_joint}
\end{equation}
where $\bm S_2 = [\bm S_{\text p},\bm S_{\text d}]$.
The joint MAP estimator is
\begin{equation}
(\hat{\bm p}_2,\{\hat\alpha_{2,k}\},\hat{\bm S}_{\text d})
=\operatorname*{arg\,min}_{{\bm p,\{\alpha_k\}
, \bm S_{\text d}\in\mathcal C_{\text d}}}
\mathcal L_2(\bm p,\{\alpha_k\},\bm S_{\text d}),
\label{eq:scenario2_map_estimator}
\end{equation}
which can be solved by alternating between data-waveform recovery
and known-waveform localization, following the joint channel-estimation and
data-detection principle in~\cite{VikHasSto:J06,KesMurMiz:C25}.

Given $(\bm p^{(t)},\{\alpha_k^{(t)}\})$, the common data update is
\begin{equation}
\hat{\bm S}_{\text d}^{(t)}\in
\arg\min_{\bm S_{\text d}\in\mathcal C_{\text d}}\sum_{k=1}^{K}
\left\|\tilde{\bm{\mathsf Y}}_{\text d,k}
-\alpha_k^{(t)}\bm B_k(\bm p^{(t)})\bm S_{\text d}
\right\|_{\text F}^2-\sigma_v^2\ln f(\bm S_{\text d}).
\label{eq:scenario2_data_update}
\end{equation}
This block can be implemented by an \ac{lmmse} estimator for Gaussian data or
by a \ac{map} detector and soft decoder for finite-alphabet data. The recovered data
waveform is then appended to the pilots as
$\hat{\bm S}^{(t)}=[\bm S_{\text p},\hat{\bm S}_{\text d}^{(t)}]$ and treated
as a known waveform. For a trial position, define
$\bm X_k(\bm p\mid\hat{\bm S}^{(t)})
\triangleq\bm B_k(\bm p)\hat{\bm S}^{(t)}$. The corresponding full-frame
amplitude estimate is
\begin{equation}
    \hat\alpha_{2,k}(\bm p\mid\hat{\bm S}^{(t)})
    =
    \frac{
    \text{tr}\big\{
    \bm X_k^{\text H}(\bm p\mid\hat{\bm S}^{(t)})
    \tilde{\bm{\mathsf Y}}_k
    \big\}}
    {\|\bm X_k(\bm p\mid\hat{\bm S}^{(t)})\|_{\text F}^2}.
    \label{eq:alpha_concentrated_est}
\end{equation}
After calculating amplitudes, the localization block becomes
\begin{equation}
    \begin{aligned}
\mathcal C_2(\bm p\mid\hat{\bm S}^{(t)})
&\triangleq\frac{1}{\sigma_v^2}\sum_{k=1}^{K}\Bigg(
\|\tilde{\bm{\mathsf Y}}_k\|_{\text F}^2-
\frac{\big|\text{tr}\big(
\bm X_k^{\text H}(\bm p\mid\hat{\bm S}^{(t)})
\tilde{\bm{\mathsf Y}}_k\big)\big|^2}
{\|\bm X_k(\bm p\mid\hat{\bm S}^{(t)})\|_{\text F}^2}\Bigg),\\
\bm p^{(t+1)}&=\arg\min_{\bm p}
\mathcal C_2(\bm p\mid\hat{\bm S}^{(t)}).
\end{aligned}
\label{eq:mstep_position}
\end{equation}
The amplitudes are then updated by~\eqref{eq:alpha_concentrated_est} at
$\bm p^{(t+1)}$. Exact block updates decrease the joint cost, while approximate
decoder updates are accepted only if they do not increase~\eqref{eq:scenario2_joint}.

\begin{remark}
The two algorithms expose the robustness-information contrast predicted by
the performance limits. Algorithm~\ref{alg:statistical_localizer} avoids symbol
reconstruction and is therefore robust to data uncertainty, but it extracts only the
covariance information available in Scenario~1. Algorithm~\ref{alg:joint_localizer}
is data-estimation aided. If $\hat{\bm S}_{\text d}^{(t)}=\bm S_{\text d}$, its
localization block is precisely the known-waveform \ac{ml} estimator associated
with Scenario~2. Otherwise, data-reconstruction errors distort the
virtual pilot signals and may propagate into the position update.
\end{remark}

\begin{figure}[t]
    \centering
    \includegraphics[width=0.88\linewidth]{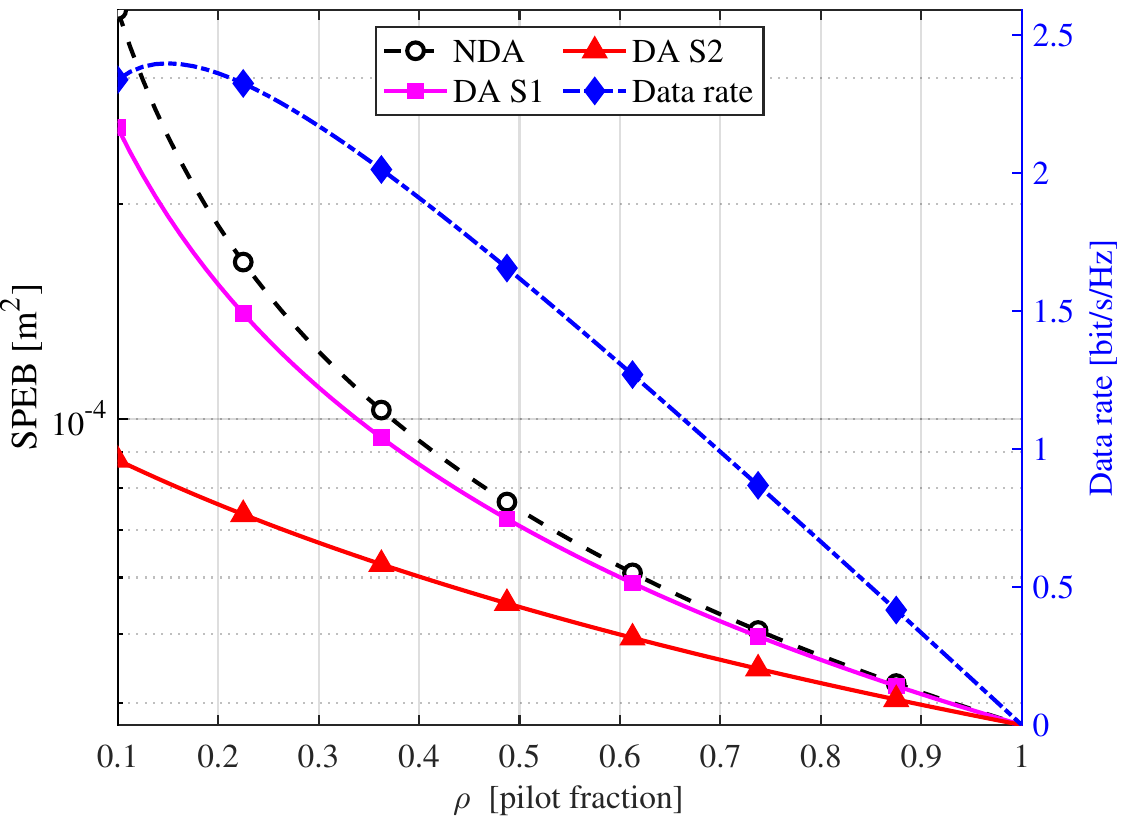}
    \caption{\ac{speb} and broadcast data rate versus pilot fraction $\rho$ at
    $5$~dB under a fixed data covariance.}
    \label{fig:theory_rho}
\end{figure}

\section{Numerical Results}
\label{sec:results}

This section validates the performance bounds for \ac{sac}, the
communication-constrained optimization, and the proposed estimators. 
Unless otherwise stated, we consider a multistatic network with one
transmitter, one point target, and $K=3$ receivers, located at
$\bm p_{\text t}=[0,0]^{\text T}$, $\bm p=[18,14]^{\text T}$,
$\bm p_{\text r,1}=[30,2]^{\text T}$, $\bm p_{\text r,2}=[5,30]^{\text T}$, and
$\bm p_{\text r,3}=[34,26]^{\text T}$ in meters. The propagation speed is
$c=3\times10^8$~m/s, the subcarrier spacing is $\Delta f=6$~MHz, and the noise
variance is normalized to $\sigma_v^2=1$. The 
$\text{SNR}=|\alpha_k|^2/\sigma_v^2$ is identical across the $K$ bistatic
links. The transmitter and each receiver employ half-wavelength \acp{ula} with
$M_{\text t}=M_{\text r}=8$ antennas, and each frame spans $N=16$ \ac{ofdm}
subcarriers and $T=80$ symbols. The pilot phase occupies $T_{\text p}$ symbols
with $M_{\text t}\le T_{\text p}\le T$ and $T_{\text d}=T-T_{\text p}$, so the
pilot fraction $\rho=T_{\text p}/T$ ranges over
$\{M_{\text t}/T,\ldots,1\}$ with minimum feasible value
$\rho_{\min}=M_{\text t}/T=0.1$.

\begin{figure}[t]
    \centering
    \includegraphics[width=0.88\linewidth]{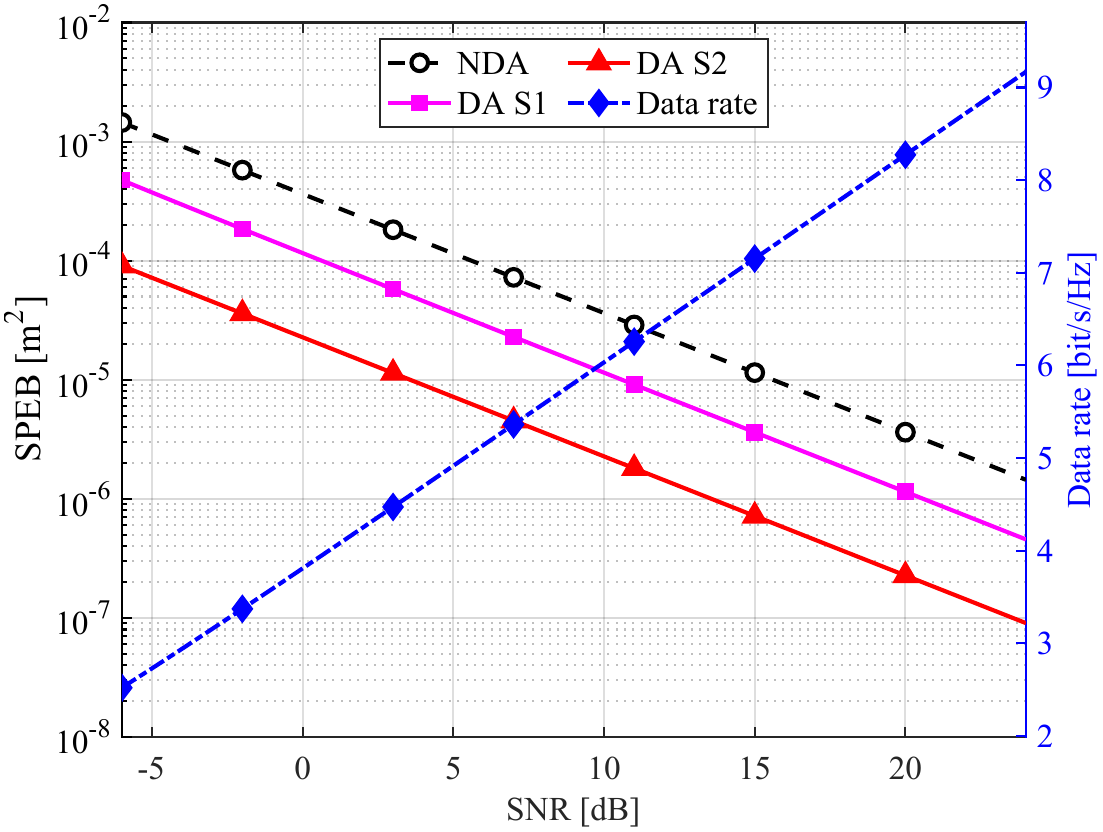}
    \caption{\ac{speb} and broadcast data rate versus \ac{snr} for NDA, DA S1,
    and DA S2 under fixed $\rho=0.32$ and a fixed $\bm R_{\text d}$.}
    \label{fig:theory_snr}
\end{figure}

\subsection{Validation of Theoretical Analysis for \ac{sac}}
To isolate the dependence of the performance limits on a single design
variable, both figures fix all other quantities and sweep only one.
Fig.~\ref{fig:theory_rho} fixes the \ac{snr} and the data covariance and
sweeps the pilot fraction $\rho$, whereas Fig.~\ref{fig:theory_snr} fixes
$\rho=0.32$ and a target-directed data covariance and sweeps the \ac{snr} to
validate the scaling under a strong reflected link.
 In all legends, NDA denotes the
non-data-aided baseline, while DA S1 and DA S2 denote data-aided Scenarios~1
and~2, respectively. The NDA bound uses only the pilot block for sensing.
Scenario~1 adds the spatial covariance information carried by the marginalized
Gaussian data, whereas Scenario~2 treats correctly decoded data as known
waveforms. The rate is evaluated from~\eqref{eq:mimo_rate_general} using the
pilot-estimated reflected channel and is common to all three  strategies
because they transmit the same waveform.

\begin{figure}[t]
    \centering
    \includegraphics[width=0.88\linewidth]{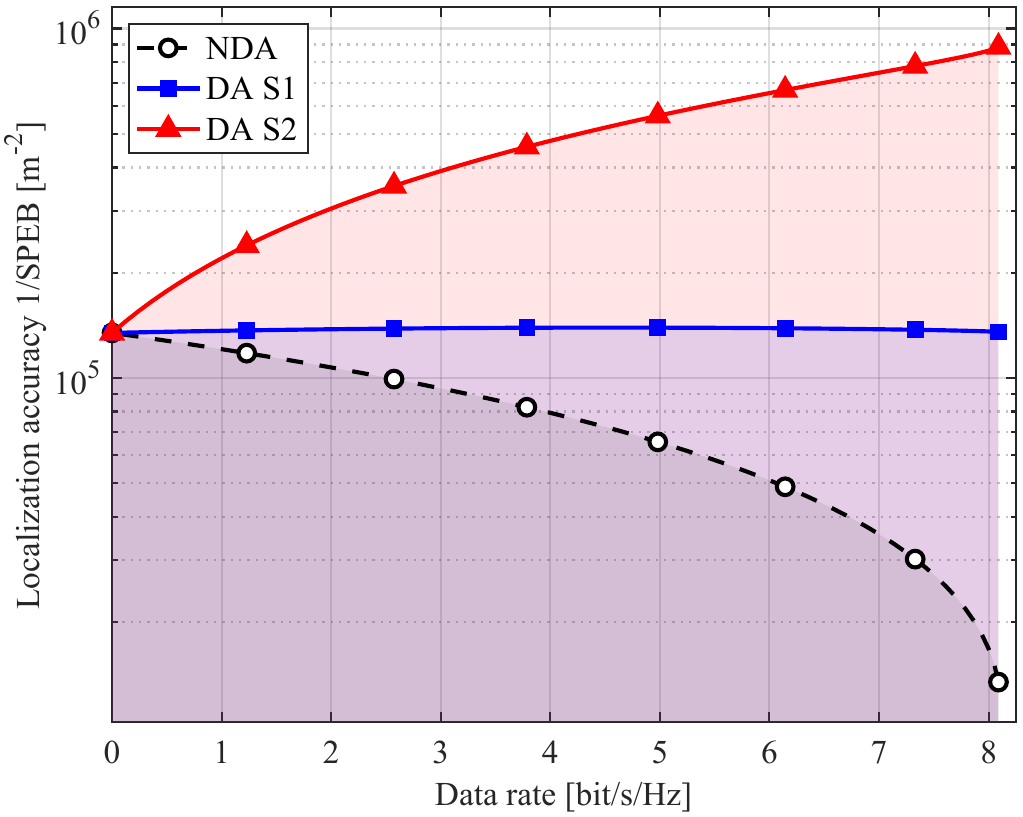}
    \caption{Joint \ac{sac} bounds obtained by sweeping the pilot fraction under
    optimized covariance at $12$~dB.}
    \label{fig:joint_sac_bound}
\end{figure}
\begin{figure}[t]
    \centering
    \includegraphics[width=0.88\linewidth]{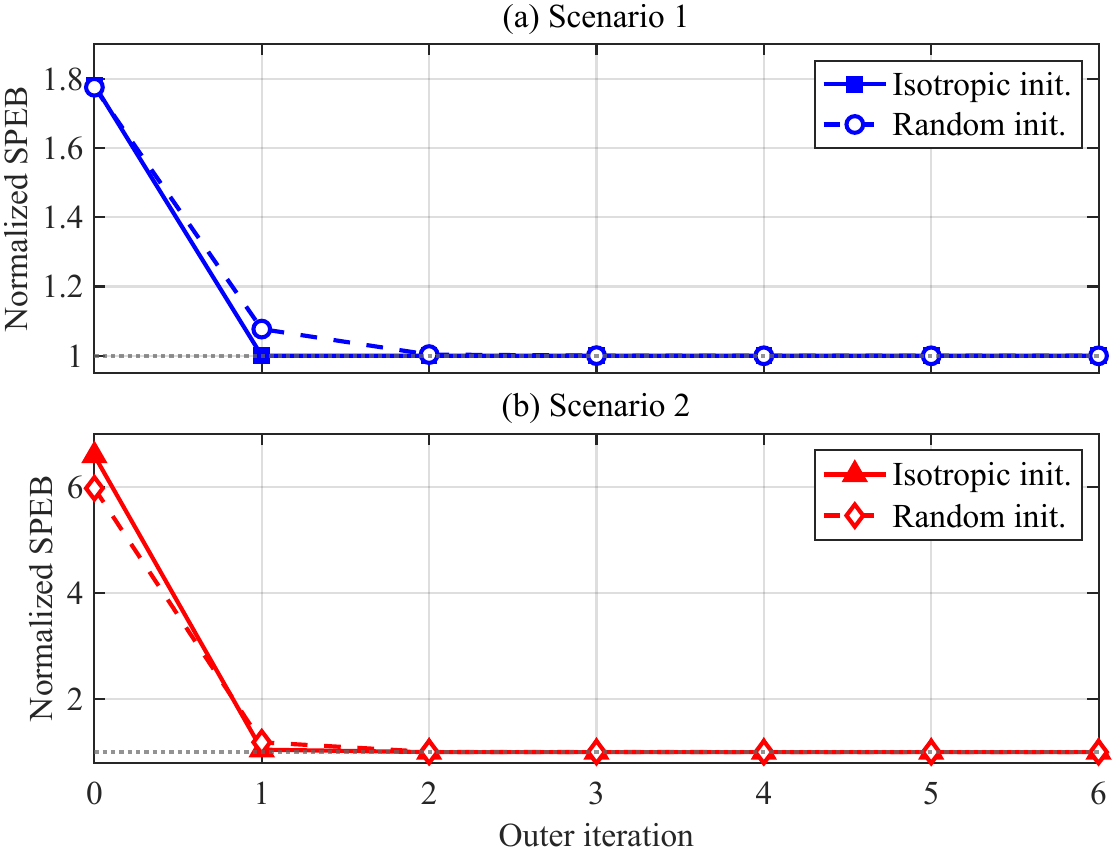}
    \caption{Convergence of Algorithm~\ref{alg:constrained_bound_opt} from isotropic and
    random feasible initializations in (a) Scenario~1 and (b) Scenario~2 at
    $R_{\text{th}}=0.6R_{\max,\text{iso}}$, normalized by the optimum
    $\mathcal P_\ell^\star$.}
    \label{fig:ao_convergence}
\end{figure}

Fig.~\ref{fig:theory_rho} fixes the \ac{snr} at $5$~dB and sweeps $\rho$. The
NDA bound decreases with $\rho$, since its localization information comes only
from the pilot block and scales with $T_{\text p}$. Scenario~1 stays below the
NDA bound at every $\rho$, because the marginalized data add strictly positive
covariance-based angular information on top of the same pilot term. This gap
narrows as $\rho\to1$, since the shrinking data phase $T_{\text d}$ carries
less covariance information. Scenario~2 attains the lowest bound over the
entire range because correctly decoded data act as virtual pilots. The
broadcast rate is maximized at an intermediate pilot fraction at $\rho\approx
0.15$ reflecting the balance between pilot-aided channel estimation and
available data duration. At $\rho=1$, the data phase vanishes, the three
\acp{speb} coincide, and the rate drops to zero, providing a direct
consistency check of the analysis.

Fig.~\ref{fig:theory_snr} fixes $\rho=0.32$ and sweeps \ac{snr}. All
\acp{speb} decrease with \ac{snr} while preserving the ordering
$\mathcal P_2<\mathcal P_1<\mathcal P_{\text{NDA}}$. The curves are nearly
parallel, indicating that data aiding provides an information gain rather than
changing the \ac{snr} order. Scenario~1 remains about half an order of magnitude
below the baseline, and Scenario~2 more than an order below it over the tested
\ac{snr} range. The broadcast rate rises from $2.5$ to $9.2$~bit/s/Hz as the
post-estimation \ac{snr} improves. Scenario~1 realizes its gain without decoding
the data realization and is statistically robust, whereas Scenario~2 is the
ideal correct-decoding benchmark and achieves the larger gain at the cost of
possible error propagation in practical receivers.

\begin{figure}[t]
    \centering
    \includegraphics[width=0.88\linewidth]{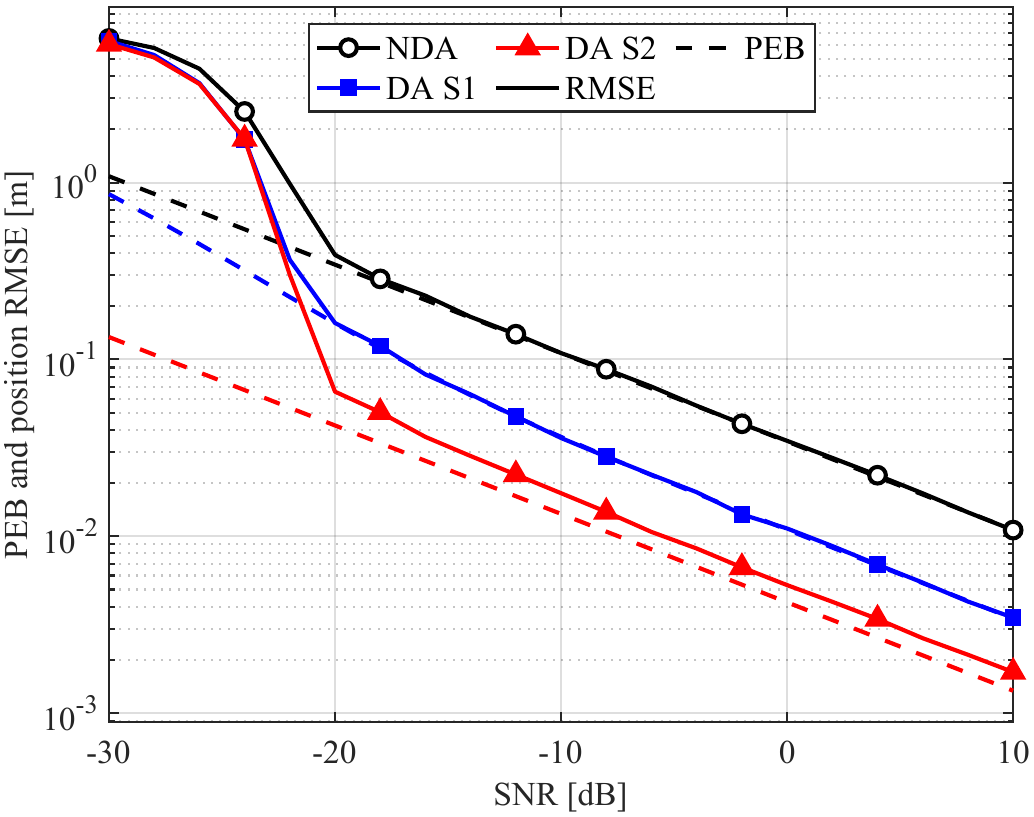}
    \caption{Position \ac{rmse} and \ac{peb} versus \ac{snr} at $\rho=0.1$ for
    the NDA, DA S1, and DA S2 localizers.}
    \label{fig:est_localization}
\end{figure}

\subsection{Validation of Communication-Constrained Optimization}
Fig.~\ref{fig:joint_sac_bound} depicts the joint \ac{sac} bound by plotting
the localization information $1/\text{SPEB}$ against the broadcast data rate at
$12$~dB, where $\bm R_{\text d}$ is fixed to the target-directed beam obtained
from Algorithm~\ref{alg:constrained_bound_opt}'s joint optimization and the
pilot fraction is swept over its feasible integer values. The gray, blue, and
red fillings denote the feasible areas below the NDA, DA S1, and DA S2 bound
curves, respectively. The NDA curve reflects the pilot-only sensing baseline and
decreases as the data rate increases, since higher rate requires a shorter pilot
phase. Scenario~1 stays comparatively flat because the loss of pilot-driven
delay information is largely compensated by data-driven angular information.
Scenario~2 gives the largest gain because reliably decoded data act as virtual
pilots and contribute to both delay and angular information.

Fig.~\ref{fig:ao_convergence} verifies the convergence of
Algorithm~\ref{alg:constrained_bound_opt} from isotropic and random 
initializations at $R_{\text{th}}=0.6R_{\max,\text{iso}}$, where
$R_{\max,\text{iso}}$ is the maximum broadcast rate attainable with an
isotropic data covariance after optimizing $\rho$ alone. The isotropic
start produces most of its decrease in the first covariance update and
stabilizes within a few outer iterations, whereas the random full-rank start
approaches the same point more gradually. Despite the different trajectories,
both starts converge to the same solution, namely $7.19\times10^{-6}$~m$^2$ in
Scenario~1 and $1.13\times10^{-6}$~m$^2$ in Scenario~2. Every accepted iterate
satisfies the broadcast-rate, positive-semidefinite, and unit-trace constraints,
with the final changes in $\rho$ and $\bm R_{\text d}$ below the prescribed
tolerances.

\subsection{Validation of Target Localization Estimator}

We finally evaluate the estimators of Section~\ref{sec:DAS} at the minimum
feasible pilot fraction $\rho=0.1$ with $T_{\text p}=8$. The data symbols are
drawn from the Gaussian codebook assumed in the analysis, and the results are
averaged over $1000$ Monte Carlo trials. All three estimators use the common
search region $[10,26]\text m\times[6,22]\text m$. Scenario~1 uses the
marginal-likelihood ML estimator in Algorithm~\ref{alg:statistical_localizer},
whereas Scenario~2 uses five alternating Gaussian-data MAP updates in
Algorithm~\ref{alg:joint_localizer}. Each accepted block update does not
increase the joint MAP cost.

\begin{figure}[t]
    \centering
    \includegraphics[width=0.88\linewidth]{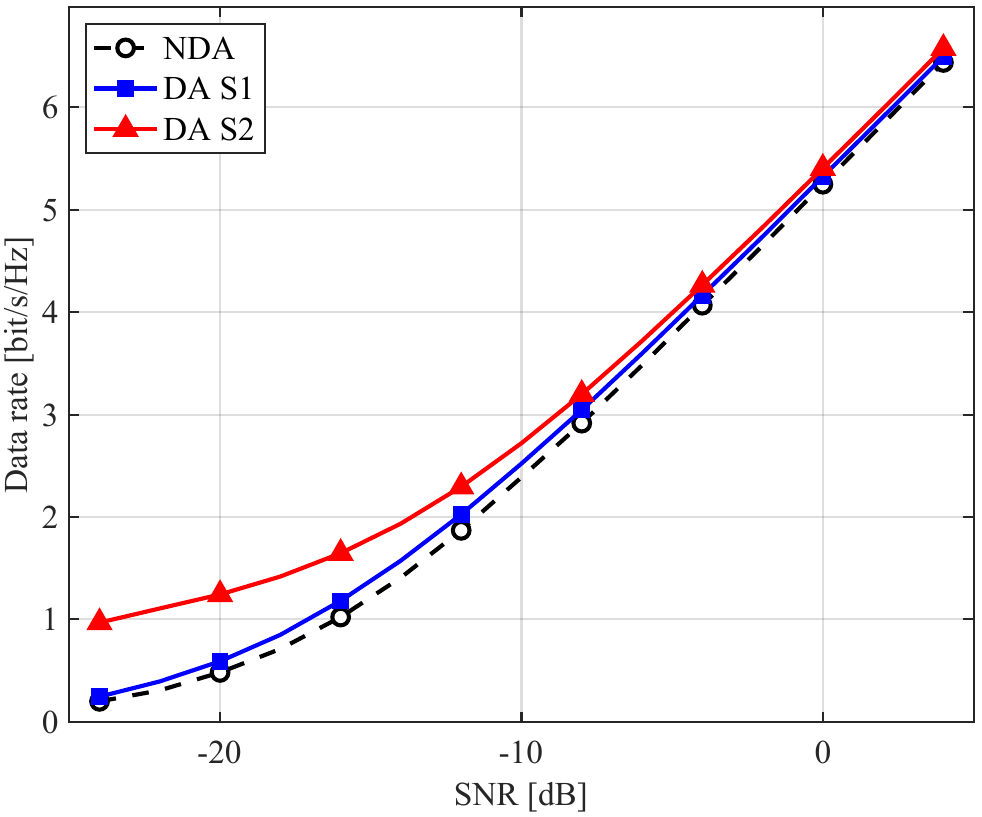}
    \caption{Broadcast data rate versus \ac{snr} at $\rho=0.1$ for  three strategies of NDA, DA S1, and DA S2.}
    \label{fig:est_rate}
\end{figure}

Fig.~\ref{fig:est_localization} shows that the practical estimators attain the
performance limits and realize the data-aided gain. Below a threshold near
$-20$~dB, the pilot initialization and Gaussian-data recovery are unreliable,
so the \acp{rmse} depart from their bounds, with Scenario~2 the most sensitive
because channel and data errors feed back through its alternating updates.
The NDA and Scenario~1 \acp{rmse} coincide with their
\acp{peb}, while the Scenario~2 \ac{rmse} stays close to but slightly above its
bound due to residual  data-estimation errors. At $10$~dB, the two
data-aided localizers reduce the position \ac{rmse} from the NDA value of
$10.8$~mm to $3.5$ and $1.7$~mm, respectively, with  \acp{peb} of
$3.4$ and $1.3$~mm.
This confirms that the gains predicted from
Theorem~\ref{thm:mimo_s1_pilot} to Theorem~\ref{thm:mimo_s2} are attained by the
algorithms.

Fig.~\ref{fig:est_rate} evaluates the complete receiver using the 
channel reconstructed by each localizer, with residual mismatch treated as
additional noise. The result shows
that data-aided localization can also improve communication by providing a more
accurate channel. Scenario~2 gives the largest gain
because decoded data act as known  information, whereas Scenario~1 improves
the reconstruction only through covariance-aided localization. The gain over the
NDA receiver is most visible at low \ac{snr}, reaching $0.97$ versus
$0.20$~bit/s/Hz at $-24$~dB and $1.65$ versus $1.02$~bit/s/Hz at $-16$~dB. It
then narrows as the pilot-only estimate becomes accurate, so the three
strategies converge at high \ac{snr}. This confirms that the data-aided
localization scheme also benefits the communication data rate by reusing the
refined channel estimate.

\section{Conclusion}\label{sec:concl}
This paper proposes a data-aided target localization framework for multistatic
\Ac{isac} networks under \ac{ofdm} signaling. Based on this, two receiver processing strategies
are considered, i.e., statistical data-aided sensing, which exploits the distribution
of unknown Gaussian data symbols, and joint data-aided sensing and decoding,
which reuses decoded data as virtual pilots. The derived performance limits for target localization show that marginalized data and decoded data provide additional sensing information from the 
 covariance and mean of the received signal, respectively. 
 Based on the
closed-form \ac{efim} and \ac{speb}, the communication-constrained localization
bound is optimized through joint time allocation and transmit-covariance design, and practical estimators
are developed to validate the two data-aided schemes.
Future work includes multi-target localization, joint pilot-data covariance optimization
under finite-alphabet constraints, and integration with iterative channel decoding.



%

\appendices

\section{Proof of Theorem~\ref{thm:mimo_s1_pilot}}
\label{app:proof_pilot_fim}
For the $k$-th link, recall the complete unknown parameter vector
$\bm\xi_k=[\tau_k,\varphi_k,\psi,\alpha_{\text R,k},
\alpha_{\text I,k}]^{\text T}$. The vectorized pilot observation is distributed
as $\tilde{\bm{\mathsf y}}_{\text p,k}\sim
\mathcal{CN}(\bm\mu_{\text p,k},\sigma_v^2\bm I)$, where
$\bm\mu_{\text p,k}=\alpha_k\text{vec}((\bm D_k\otimes\bm A(\bm\vartheta_k))
\bm S_{\text p})$. Since the covariance is independent
of $\bm\xi_k$, the complex Gaussian \ac{fim}~\cite{KhaPodHaa:J21} of
$\bm \xi_k$ gives
\begin{equation}
    \bm J_{\text p,k}(\bm\xi_k)
    =
    \frac{2}{\sigma_v^2}\text{Re}\Big\{
    \frac{\partial\bm\mu_{\text p,k}^{\text H}}{\partial \bm\xi_{k}}
    \frac{\partial\bm\mu_{\text p,k}}{\partial \bm \xi_{k}^\text{T}}
    \Big\}.
    \label{eq:app_pilot_fim_definition}
\end{equation}
For brevity, write $\bm D_k\triangleq\bm D(\tau_k)$. Define
$\bm B_{0,k}=\bm D_k\otimes\bm A(\bm\vartheta_k)$,
$\bm B_{\tau,k}=\dot{\bm D}_k\otimes\bm A(\bm\vartheta_k)$,
$\bm B_{\varphi,k}=\bm D_k\otimes\dot{\bm A}_{\varphi,k}$, and
$\bm B_{\psi,k}=\bm D_k\otimes\dot{\bm A}_{\psi,k}$, where
$\dot{\bm D}_k=-j2\pi\Delta f\,
\text{diag}(0,\ldots,N-1)\bm D_k$,
$\dot{\bm A}_{\varphi,k}=\dot{\bm a}_{\text r}(\varphi_k)
\bm a_{\text t}^{\text H}(\psi)$, and
$\dot{\bm A}_{\psi,k}=\bm a_{\text r}(\varphi_k)
\dot{\bm a}_{\text t}^{\text H}(\psi)$.
The mean derivatives with respect to $i\in\{\tau,\varphi,\psi\}$ are
$\alpha_k\text{vec}(\bm B_{i,k}\bm S_{\text p})$, while those with respect
to $(\alpha_{\text R,k},\alpha_{\text I,k})$ are
$\text{vec}(\bm B_{0,k}\bm S_{\text p})$ and
$j\text{vec}(\bm B_{0,k}\bm S_{\text p})$, respectively.
For $i,j\in\{0,\tau,\varphi,\psi\}$, define
$c_{ij,k}\triangleq
\text{vec}(\bm B_{i,k}\bm S_{\text p})^{\text H}
\text{vec}(\bm B_{j,k}\bm S_{\text p})$ and $c_{0,k}\triangleq c_{00,k}$.
Substitution into \eqref{eq:app_pilot_fim_definition} yields the block form
\begin{equation}
\bm J_{\text p,k}(\bm\xi_k)
=\frac{2}{\sigma_v^2}
\begin{bmatrix}
|\alpha_k|^2\bm C_{\text p,k} & \bm D_{\text p,k}\\
\bm D_{\text p,k}^{\text T} & c_{0,k}\bm I_2
\end{bmatrix},
\label{eq:app_pilot_full_fim}
\end{equation}
where $[\bm C_{\text p,k}]_{ij}=\text{Re}\{c_{ij,k}\}$ for
$i,j\in\{\tau,\varphi,\psi\}$, and the $i$-th row of
$\bm D_{\text p,k}$ is
$[\text{Re}\{\alpha_k^*c_{i0,k}\},-\text{Im}\{\alpha_k^*c_{i0,k}\}]$.
We next evaluate these entries under the per-subcarrier orthogonal-pilot
condition~\eqref{eq:orthogonal_pilot_main}. Since
$\bm B_{i,k}^{\text H}\bm B_{j,k}$ is block diagonal across subcarriers, only
the diagonal subcarrier blocks of
$\bm S_{\text p}\bm S_{\text p}^{\text H}$ enter the trace. Hence,
$\text{vec}(\bm B_{i,k}\bm S_{\text p})^{\text H}
 \text{vec}(\bm B_{j,k}\bm S_{\text p})
=(T_{\text p}/M_{\text t})
 \text{tr}(\bm B_{i,k}^{\text H}\bm B_{j,k})$.
Consequently, $c_{ij,k}
=(T_{\text p}/M_{\text t})
\text{tr}(\bm B_{i,k}^{\text H}\bm B_{j,k})$, whose Kronecker factors can be
evaluated separately. For a \ac{ula}, we have 
\begin{equation}
\begin{aligned}
    \bm a_q(\theta_q)
    &=\big[e^{jm u_q(\theta_q)}\big]_{m=0}^{M_q-1},\\
    u_q(\theta_q)
    &\triangleq 2\pi\Delta_q\sin\theta_q,~~~~
    u_q'(\theta_q)=2\pi\Delta_q\cos\theta_q,\\
    \dot{\bm a}_q(\theta_q)
    &=j u_q'(\theta_q)\text{diag}(0,\ldots,M_q-1)\bm a_q(\theta_q),
    \quad q\in\{\text r,\text t\}.
\end{aligned}
    \label{eq:app_ula_derivative}
\end{equation}
Here, $\theta_{\text r}=\varphi_k$, $\theta_{\text t}=\psi$, and
$\Delta_q\triangleq d_q/\lambda$ is the inter-element spacing of array $q$
normalized by the carrier wavelength. For half-wavelength spacing,
$\Delta_q=1/2$ and $u_q'(\theta_q)=\pi\cos\theta_q$.
Eliminating $(\alpha_{\text R,k},\alpha_{\text I,k})$ from
\eqref{eq:app_pilot_full_fim} by the Schur complement gives
$\bm J_{\text p}^{\text e}(\bm\theta_k)$~\cite{SheWin:J10a}. Using
the trace identity above, the off-diagonal terms cancel and the diagonal terms
retain only the centered frequency and array apertures. With
$s_\ell(L)\triangleq\sum_{m=0}^{L-1}m^\ell$ for $\ell\in\{1,2\}$,
$\omega\triangleq 2\pi\Delta f$,
$\nu_{\text r}\triangleq u_{\text r}'(\varphi_k)$, and
$\nu_{\text t}\triangleq u_{\text t}'(\psi)$, this gives
\begin{equation}
\begin{aligned}
\bm J_{\text p}^{\text e}(\bm\theta_k)
&=\frac{2T_{\text p}|\alpha_k|^2}{\sigma_v^2}
\text{diag}\Big\{
M_{\text r}\omega^2
\big[s_2(N)-\frac{s_1^2(N)}{N}\big],\\
&\quad
N\nu_{\text r}^2
\big[s_2(M_{\text r})-\frac{s_1^2(M_{\text r})}{M_{\text r}}\big],\\
&\quad
\frac{NM_{\text r}}{M_{\text t}}\nu_{\text t}^2
\big[s_2(M_{\text t})-\frac{s_1^2(M_{\text t})}{M_{\text t}}\big]
\Big\}.
\end{aligned}
\label{eq:app_pilot_final_efim}
\end{equation}
This is exactly the centered-aperture form in~\eqref{eq:mimo_pilot_diag}.
Using $s_2(L)-s_1^2(L)/L=L(L^2-1)/12$ gives the equivalent fully expanded
form.
\hfill$\square$

\section{Proof of Lemma~\ref{thm:mimo_s1_data}}
\label{app:proof_data_fim}
After marginalizing the Gaussian data, each data observation is zero-mean
Gaussian with covariance $\bm I_N\otimes\tilde{\bm Q}_k$, where
$\tilde{\bm Q}_k$ is given in Lemma~\ref{thm:mimo_s1_data}. The data covariance  is calculated by 
$|\alpha_k|^2(\bm D_k\otimes\bm A(\bm\vartheta_k))(\bm I_N\otimes\bm R_{\text d})
(\bm D_k\otimes\bm A(\bm\vartheta_k))^{\text H}+\sigma_v^2\bm I_{NM_{\text r}}$.
Since $\bm D_k\bm D_k^{\text H}=\bm I_N$, this reduces to
$\bm\Sigma_k=\bm I_N\otimes\tilde{\bm Q}_k$ and is independent of $\tau_k$.
Hence $\partial\bm\Sigma_k/\partial\tau_k=\bm0$ and
$[\bm J_{\text d,k}]_{\tau i}=0$ for
$i\in\{\tau,\varphi,\psi,\alpha_{\text R,k},\alpha_{\text I,k}\}$.
The zero first row and column in \eqref{eq:main_data_theta_fim} therefore come
from the marginalization of the unknown data symbols, which removes the
subcarrier phase rotations. The complete covariance over one data symbol is
$\bm I_N\otimes\tilde{\bm Q}_k$, and over $T_{\text d}$ independent data symbols it
is repeated $T_{\text d}$ times. For zero-mean complex Gaussian
observations~\cite{KhaPodHaa:J21},
\begin{equation}
    [\bm J_{\text d,k}]_{ij}
    =
    NT_{\text d}\text{tr}\big(
    \tilde{\bm Q}_k^{-1}
    \dot{\tilde{\bm Q}}_{i,k}
    \tilde{\bm Q}_k^{-1}
    \dot{\tilde{\bm Q}}_{j,k}
    \big).
    \label{eq:app_cov_fim}
\end{equation}
The angular covariance derivatives are
$    \dot{\tilde{\bm Q}}_{i,k}
    =
    \dot{\bm H}_{i,k}\bm R_{\text d}\bm H_k^{\text H}
    +
    \bm H_k\bm R_{\text d}\dot{\bm H}_{i,k}^{\text H},
    ~ i\in\{\varphi,\psi\}$,
with
$\dot{\bm H}_{\varphi,k}=\alpha_k\dot{\bm a}_{\text r}(\varphi_k)\bm a_{\text t}^{\text H}(\psi)$
and
$\dot{\bm H}_{\psi,k}=\alpha_k\bm a_{\text r}(\varphi_k)\dot{\bm a}_{\text t}^{\text H}(\psi)$,
which lead to the angular block in \eqref{eq:main_data_theta_fim}.
Using the covariance inner product defined in
\eqref{eq:main_data_gamma_def}, \eqref{eq:app_cov_fim} gives
$[\bm J_{\text d,k}]_{ij}=T_{\text d}[\bm\Gamma_{\text d,k}]_{ij}$ for
$i,j\in\{\varphi,\psi\}$.
For the amplitude blocks, let
$\bm C_k=\bm A(\bm\vartheta_k)\bm R_{\text d}\bm A^{\text H}(\bm\vartheta_k)$. Since
$\tilde{\bm Q}_k=|\alpha_k|^2\bm C_k+\sigma_v^2\bm I$,
$\partial\tilde{\bm Q}_k/\partial\alpha_{\text R,k}
=2\alpha_{\text R,k}\bm C_k$ and
$\partial\tilde{\bm Q}_k/\partial\alpha_{\text I,k}
=2\alpha_{\text I,k}\bm C_k$. Combining these derivatives with the angular
covariance derivative gives
\begin{equation}
    [\bm J_{\text d,k}]_{i\alpha}
    =
    2NT_{\text d}
    \text{tr}(\tilde{\bm Q}_k^{-1}\dot{\tilde{\bm Q}}_{i,k}
    \tilde{\bm Q}_k^{-1}\bm C_k)
    \bm\alpha_k^{\text T},
    \quad i\in\{\varphi,\psi\},
    \label{eq:app_data_angle_amp}
\end{equation}
Using \eqref{eq:main_data_gamma_def}, this is equivalently
$[\bm J_{\text d,k}]_{i\alpha}
=T_{\text d}[\bm b_{\text d,k}]_i\bm\alpha_k^{\text T}$ for
$i\in\{\varphi,\psi\}$,
which gives the data angle-amplitude block in \eqref{eq:main_data_theta_fim}.
Similarly, with
$\bm\alpha_k=[\alpha_{\text R,k},\alpha_{\text I,k}]^{\text T}$,
\begin{equation}
    \bm J_{\text d,\alpha\alpha,k}
    =
    4NT_{\text d}
    \text{tr}(\tilde{\bm Q}_k^{-1}\bm C_k
    \tilde{\bm Q}_k^{-1}\bm C_k)
    \bm\alpha_k\bm\alpha_k^{\text T},
    \label{eq:app_data_amp_rank_one}
\end{equation}
or, using $\zeta_k=4\langle\bm C_k,\bm C_k\rangle_k$,
$\bm J_{\text d,\alpha\alpha,k}
=T_{\text d}\zeta_k\bm\alpha_k\bm\alpha_k^{\text T}$.
This block is rank one and points along the radial direction $\bm\alpha_k$ in the
complex-amplitude plane, because the marginalized data covariance depends on
the complex gain only through $|\alpha_k|^2$.
Collecting the zero delay row and column, the angular block, the
angle-amplitude block, and the amplitude block gives
\begin{equation}
    \bm J_{\text d}(\bm \xi_k)
    =
    T_{\text d}
    \begin{bmatrix}
    0 & \bm0_{1\times2} & \bm0_{1\times2}\\
    \bm0_{2\times1} & \bm\Gamma_{\text d,k}
    & \bm b_{\text d,k}\bm \alpha_k^{\text T}\\
    \bm0_{2\times1} & \bm \alpha_k\bm b_{\text d,k}^{\text T}
    & \zeta_k\bm \alpha_k\bm \alpha_k^{\text T}
    \end{bmatrix},
    \label{eq:app_data_final_block}
\end{equation}
which gives \eqref{eq:main_data_theta_fim}. \hfill$\square$

\section{Proof of Theorem~\ref{thm:mimo_s1}}
\label{app:proof_joint_efim}
For the $k$-th link, recall
$\bm\xi_k=[\bm\theta_k^{\text T},\bm\alpha_k^{\text T}]^{\text T}$ with
$\bm\theta_k=[\tau_k,\varphi_k,\psi]^{\text T}$. Because the complex amplitude
$\bm\alpha_k$ is shared by the pilot and data observations, the total \ac{fim}
is the sum $\bm J_k=\bm J_{\text p,k}+\bm J_{\text d,k}$,
with $\bm J_{\text p,k}$ in \eqref{eq:app_pilot_full_fim} and $\bm J_{\text d,k}$
in \eqref{eq:main_data_theta_fim}. Partition the total \ac{fim}
$\bm J_k(\bm\xi_k)$ as
   $ \bm J_k(\bm\xi_k)
    =
    \begin{bmatrix}
    \bm J_{\theta\theta,k} & \bm J_{\theta\alpha,k}\\
    \bm J_{\alpha\theta,k} & \bm J_{\alpha\alpha,k}
    \end{bmatrix}$.
The amplitude is therefore eliminated only after the addition, through the
single Schur complement~\cite{SheWin:J10a}
\begin{equation}
    \bm J_1^{\text e}(\bm\theta_k)
    =\bm J_{\theta\theta,k}
    -\bm J_{\theta\alpha,k}\bm J_{\alpha\alpha,k}^{-1}\bm J_{\alpha\theta,k}.
    \label{eq:app_joint_schur}
\end{equation}
The key point is that \eqref{eq:app_joint_schur} must be applied to the
{sum} of the pilot and data FIMs, because both observations contain the
same nuisance amplitude $\bm\alpha_k$. The remaining calculation only involves
the two real amplitude variables
$(\alpha_{\text R,k},\alpha_{\text I,k})$. We therefore use the following
orthonormal basis in this two-dimensional nuisance-amplitude space.
\begin{equation}
\begin{aligned}
    \bm e_{\text r,k}
    &\triangleq
    \frac{1}{\|\bm\alpha_k\|}
    [\alpha_{\text R,k},\alpha_{\text I,k}]^{\text T},\\
    \bm e_{\text t,k}
    &\triangleq
    \frac{1}{\|\bm\alpha_k\|}
    [\alpha_{\text I,k},-\alpha_{\text R,k}]^{\text T},
    \qquad
    \bm e_{\text r,k}^{\text T}\bm e_{\text t,k}=0.
\end{aligned}
    \label{eq:app_radial_tangential}
\end{equation}
Here $\bm e_{\text r,k}$ is the direction in which $|\alpha_k|$ changes, whereas
$\bm e_{\text t,k}$ is the orthogonal direction in which only the phase of
$\alpha_k$ changes. This is only an orthonormal change of coordinates in the
amplitude block of \eqref{eq:app_joint_schur}.

By the definition of $c_{ij,k}$, $c_{0,k}=\|\text{vec}(\bm B_{0,k}\bm
S_{\text p})\|^2=T_{\text p}NM_{\text r}$ is independent of
$\bm\theta_k$. Hence $\text{Re}\{c_{i0,k}\}=0$ for
$i\in\{\tau,\varphi,\psi\}$, and the pilot parameter--amplitude couplings in
\eqref{eq:app_pilot_full_fim} are tangential to the complex amplitude, with
$\bm J_{\text p,i\alpha}=p_{i,k}\bm e_{\text t,k}^{\text T}$ for
$i\in\{\tau,\varphi,\psi\}$ and real coefficients $p_{i,k}$. The pilot
amplitude block is isotropic,
$\bm J_{\text p,\alpha\alpha}=j_{\text p,k}^{\alpha}\bm I_2$, with
$j_{\text p,k}^{\alpha}=2T_{\text p}NM_{\text r}/\sigma_v^2$. Therefore the
pilot-only Schur complement satisfies
\begin{equation}
    [\bm J_{\text p}^{\text e}(\bm\theta_k)]_{i,i}
    =
    [\bm J_{\text p,\theta\theta,k}]_{i,i}
    -\frac{p_{i,k}^2}{j_{\text p,k}^{\alpha}}
    =
    T_{\text p}\gamma_{\text p,k}^i .
    \label{eq:app_pilot_scalar_schur}
\end{equation}

To expose the data structure, let
$\beta_k=\bm a_{\text t}^{\text H}(\psi)\bm R_{\text d}
\bm a_{\text t}(\psi)$, so that
$\bm C_k=\beta_k\bm a_{\text r}(\varphi_k)
\bm a_{\text r}^{\text H}(\varphi_k)$. The constant-modulus receive response
implies $[\bm b_{\text d,k}]_\varphi=0$ and
$[\bm\Gamma_{\text d,k}]_{\varphi,\psi}=0$
because $\bm a_{\text r}^{\text H}\dot{\bm a}_{\text r}$ is purely imaginary,
whereas $\tilde{\bm Q}_k^{-1}\bm a_{\text r}$ is collinear with
$\bm a_{\text r}$. Thus the marginalized data have neither a delay term nor an
\ac{aoa}--amplitude coupling. Their only possible radial angular coupling is
$\bm J_{\text d,\psi\alpha}=q_{\psi,k}\bm e_{\text r,k}^{\text T}$, where
$q_{\psi,k}=T_{\text d}[\bm b_{\text d,k}]_\psi\|\bm\alpha_k\|$.
The data amplitude block in \eqref{eq:app_data_amp_rank_one} is
$T_{\text d}\zeta_k\|\bm\alpha_k\|^2
\bm e_{\text r,k}\bm e_{\text r,k}^{\text T}$. The joint amplitude block is
$\bm J_{\alpha\alpha,k}
=j_{\text p,k}^\alpha\bm I_2
+T_{\text d}\zeta_k\|\bm\alpha_k\|^2
\bm e_{\text r,k}\bm e_{\text r,k}^{\text T}$,
which is diagonal in $(\bm e_{\text r,k},\bm e_{\text t,k})$. Hence
\begin{equation}
    \bm J_{\alpha\alpha,k}^{-1}
    =
    \frac{\bm e_{\text r,k}\bm e_{\text r,k}^{\text T}}
    {j_{\text p,k}^{\alpha}+T_{\text d}\zeta_k\|\bm\alpha_k\|^2}
    +
    \frac{\bm e_{\text t,k}\bm e_{\text t,k}^{\text T}}
    {j_{\text p,k}^{\alpha}} .
    \label{eq:app_joint_amp_inverse}
\end{equation}

The pilot and data angle-amplitude couplings lie in orthogonal amplitude
directions. Since the pilot-only EFIM is diagonal and
$[\bm\Gamma_{\text d,k}]_{\varphi,\psi}=0$, all off-diagonal entries of the
joint Schur complement vanish. For the diagonal entries, the delay has no data
term and the \ac{aoa} has no radial data coupling, so
$\lambda_k^\tau=T_{\text p}\gamma_{\text p,k}^\tau$ and
$\lambda_k^\varphi=T_{\text p}\gamma_{\text p,k}^\varphi
+T_{\text d}[\bm\Gamma_{\text d,k}]_{\varphi,\varphi}$. For the \ac{aod}, the
total amplitude coupling is
\begin{equation}
    \bm J_{\psi\alpha,k}
    =
    p_{\psi,k}\bm e_{\text t,k}^{\text T}
    +q_{\psi,k}\bm e_{\text r,k}^{\text T}.
    \label{eq:app_psi_amp_coupling}
\end{equation}
Substituting \eqref{eq:app_joint_amp_inverse} into the Schur-complement
correction gives
\begin{equation}
    [\bm J_{\theta\alpha,k}\bm J_{\alpha\alpha,k}^{-1}
    \bm J_{\alpha\theta,k}]_{\psi\psi}
    =\frac{p_{\psi,k}^2}{j_{\text p,k}^\alpha}
    +\frac{q_{\psi,k}^2}{j_{\text p,k}^\alpha
    +T_{\text d}\zeta_k\|\bm\alpha_k\|^2}.
    \label{eq:app_psi_schur}
\end{equation}
Combining \eqref{eq:app_pilot_scalar_schur} and \eqref{eq:app_psi_schur},
\begin{equation}
\lambda_k^\psi
=
T_{\text p}\gamma_{\text p,k}^\psi
+T_{\text d}\left(
[\bm\Gamma_{\text d,k}]_{\psi,\psi}
-\frac{T_{\text d}[\bm b_{\text d,k}]_\psi^2\|\bm\alpha_k\|^2}
{j_{\text p,k}^{\alpha}+T_{\text d}\zeta_k\|\bm\alpha_k\|^2}
\right),
    \label{eq:app_lambda_psi_final}
\end{equation}
which is exactly \eqref{eq:mimo_gamma_data}. As $\rho$ increases,
$j_{\text p,k}^\alpha\propto T_{\text p}$ dominates and the residual correction
vanishes. This proves \eqref{eq:main_total_schur}--\eqref{eq:mimo_gamma_data}.
\hfill$\square$

\section{Proof of Theorem~\ref{thm:mimo_s2}}
\label{app:proof_s2_efim}
Conditioned on a correctly decoded data realization~\cite{ZhaKamAlo:J25}, the full waveform
$\bm S_2=[\bm S_{\text p},\bm S_{\text d}]$ is deterministic and known. Hence the
mean-based \ac{fim} is identical to the pilot \ac{fim} in
Appendix~\ref{app:proof_pilot_fim}, with $\bm S_{\text p}$ replaced by
$\bm S_2$. For the expected benchmark used for offline covariance design, let
$\bm{\mathsf S}_2=[\bm S_{\text p},\bm{\mathsf S}_{\text d}]$ denote the
random full-frame waveform. Averaging its Gram matrix gives
\begin{equation}
\mathbb E_{\bm{\mathsf S}_{\text d}}
\{\bm{\mathsf S}_2\bm{\mathsf S}_2^{\text H}\}
=T(\bm I_N\otimes\bm R_2).
\label{eq:app_s2_gram}
\end{equation}
Accordingly,
\begin{equation}
\bar c_{ij,k}
=T\operatorname{tr}(
\bm B_{i,k}^{\text H}\bm B_{j,k}
(\bm I_N\otimes\bm R_2)).
\label{eq:app_s2_cbar}
\end{equation}
Substituting this into the full FIM \eqref{eq:app_pilot_full_fim} and
eliminating the common complex amplitude by the same Schur complement give
\begin{equation}
[\bar{\bm J}_{2}^{\text e}]_{ij}
=\frac{2T|\alpha_k|^2}{\sigma_v^2}
\operatorname{Re}\left\{
\chi_{ij,k}
-\frac{\chi_{i0,k}\chi_{0j,k}}{\chi_{00,k}}
\right\},
\label{eq:app_s2_schur_entry}
\end{equation}
where
\begin{equation}
\chi_{ij,k}\triangleq\frac{\bar c_{ij,k}}{T}.
\label{eq:app_s2_chi_def}
\end{equation}
The Kronecker factors separate the
frequency, receive-array, and transmit-array terms. Their cross terms cancel,
and the three diagonal terms reduce to \eqref{eq:mimo_s2_lambda} using
\eqref{eq:mimo_s2_moments}. The position transformation of
Corollary~\ref{thm:mimo_s1_position} then gives \eqref{eq:mimo_s2_efim}.
\hfill$\square$

\bibliographystyle{IEEEtran}
\bibliography{IEEEAbrv,StringDefinitions,SGroupDefinition,SGroup}

\begin{thebibliography}{10}
\providecommand{\url}[1]{#1}
\csname url@samestyle\endcsname
\providecommand{\newblock}{\relax}
\providecommand{\bibinfo}[2]{#2}
\providecommand{\BIBentrySTDinterwordspacing}{\spaceskip=0pt\relax}
\providecommand{\BIBentryALTinterwordstretchfactor}{4}
\providecommand{\BIBentryALTinterwordspacing}{\spaceskip=\fontdimen2\font plus
\BIBentryALTinterwordstretchfactor\fontdimen3\font minus
  \fontdimen4\font\relax}
\providecommand{\BIBforeignlanguage}[2]{{%
\expandafter\ifx\csname l@#1\endcsname\relax
\typeout{** WARNING: IEEEtran.bst: No hyphenation pattern has been}%
\typeout{** loaded for the language `#1'. Using the pattern for}%
\typeout{** the default language instead.}%
\else
\language=\csname l@#1\endcsname
\fi
#2}}
\providecommand{\BIBdecl}{\relax}
\BIBdecl

\bibitem{LiuLiuCui:J25}
F.~Liu, Y.-F. Liu, Y.~Cui, C.~Masouros, J.~Xu, T.~X. Han, S.~Buzzi, Y.~C.
  Eldar, and S.~Jin, ``\textnormal{Sensing with communication signals: From
  information theory to signal processing},'' \emph{{IEEE} J. Sel. Areas
  Commun.}, vol.~44, pp. 1--30, Oct. 2025.

\bibitem{LiuZhaZha:J24}
F.~Liu, T.~Zhang, Z.~Zhang, Y.~Shen, and Q.~Zhang, ``\textnormal{ISAC with UWB:
  Reliable decoupling and target sensing},'' \emph{{IEEE} Trans. Wireless
  Commun.}, vol.~23, pp. 15\,957--15\,972, Nov. 2024.

\bibitem{TanYuPan:J25}
J.~Tang, Y.~Yu, C.~Pan, H.~Ren, D.~Wang, J.~Wang, and X.~You,
  ``\textnormal{Cooperative ISAC-empowered low-altitude economy},''
  \emph{{IEEE} Trans. Wireless Commun.}, vol.~24, no.~5, pp. 3837--3853, May
  2025.

\bibitem{WanWuShe:J20}
Y.~Wang, Y.~Wu, and Y.~Shen, ``\textnormal{Cooperative tracking by multi-agent
  systems using signals of opportunity},'' \emph{{IEEE} Trans. Commun.},
  vol.~68, no.~1, pp. 93--105, Jan. 2020.

\bibitem{DuLiuLi:J25}
Z.~Du, F.~Liu, Y.~Li, W.~Yuan, Y.~Cui, and Z.~Zhang, ``\textnormal{Toward
  ISAC-empowered vehicular networks: Framework, advances, and opportunities},''
  \emph{{IEEE} Wireless Commun. Mag.}, vol.~32, no.~2, pp. 222--229, Apr. 2025.

\bibitem{LiuMasLiSunHan:J18}
F.~Liu, C.~Masouros, A.~Li, H.~Sun, and L.~Hanzo, ``\textnormal{MU-MIMO
  communications with MIMO radar: From co-existence to joint transmission},''
  \emph{{IEEE} Trans. Wireless Commun.}, vol.~17, no.~4, pp. 2755--2770, Feb.
  2018.

\bibitem{XioLiuCui:J22}
Y.~Xiong, F.~Liu, Y.~Cui, W.~Yuan, T.~X. Han, and G.~Caire, ``\textnormal{On
  the fundamental tradeoff of integrated sensing and communications under
  Gaussian channels},'' \emph{{IEEE} Trans. Inf. Theory}, vol.~69, no.~9, pp.
  5723--5751, Sep. 2023.

\bibitem{SheLuZha:J26}
X.~Shen, Z.~Lu, N.~Zhao, H.~Zhao, and Y.~Shen, ``\textnormal{Fundamental
  tradeoff of bistatic ISAC under Gaussian fading channels at finite
  blocklength},'' \emph{{IEEE} Trans. Inf. Theory}, vol.~72, no.~2, pp.
  1176--1200, Feb. 2026.

\bibitem{KesMojLac:J25}
M.~F. Keskin, M.~M. Mojahedian, J.~O. Lacruz, C.~Marcus, O.~Eriksson,
  A.~Giorgetti, J.~Widmer, and H.~Wymeersch, ``\textnormal{Fundamental
  trade-offs in monostatic ISAC: A holistic investigation toward 6G},''
  \emph{{IEEE} Trans. Wireless Commun.}, vol.~24, pp. 7856--7873, Sep. 2025.

\bibitem{LiuHuaShlLiuZhoEld:J20}
X.~Liu, T.~Huang, N.~Shlezinger, Y.~Liu, J.~Zhou, and Y.~C. Eldar,
  ``\textnormal{Joint transmit beamforming for multiuser MIMO communications
  and MIMO radar},'' \emph{{IEEE} Trans. Signal Process.}, vol.~68, pp.
  3929--3944, Jun. 2020.

\bibitem{ZhaWanZha:J21}
N.~Zhao, Y.~Wang, Z.~Zhang, Q.~Chang, and Y.~Shen, ``\textnormal{Joint transmit
  and receive beamforming design for integrated sensing and communication},''
  \emph{{IEEE} Commun. Lett.}, vol.~26, no.~3, pp. 662--666, Jan. 2022.

\bibitem{DuLiuXio:J24}
Z.~Du, F.~Liu, Y.~Xiong, T.~X. Han, Y.~C. Eldar, and S.~Jin,
  ``\textnormal{Reshaping the ISAC tradeoff under OFDM signaling: A
  probabilistic constellation shaping approach},'' \emph{{IEEE} Trans. Signal
  Process.}, vol.~72, pp. 4782--4797, Sep. 2024.

\bibitem{LiuZhaXio:J25}
F.~Liu, Y.~Zhang, Y.~Xiong, S.~Li, W.~Yuan, F.~Gao, S.~Jin, and G.~Caire,
  ``\textnormal{CP-OFDM achieves the lowest average ranging sidelobe under
  QAM/PSK constellations},'' \emph{{IEEE} Trans. Inf. Theory}, vol.~71, no.~9,
  pp. 6950--6968, Sep. 2025.

\bibitem{ZhaChaShe:J25}
N.~Zhao, Q.~Chang, X.~Shen, Y.~Wang, and Y.~Shen, ``\textnormal{Joint target
  localization and data detection in bistatic ISAC networks},'' \emph{{IEEE}
  Trans. Commun.}, vol.~73, no.~5, pp. 3531--3546, May 2025.

\bibitem{WuLiHe:J25}
N.~Wu, H.~Li, D.~He, A.~Nallanathan, and T.~Q.~S. Quek,
  ``\textnormal{Integrated sensing and communication receiver design for
  OTFS-based MIMO system: A unified variational inference framework},''
  \emph{{IEEE} J. Sel. Areas Commun.}, vol.~43, no.~4, pp. 1339--1353, Apr.
  2025.

\bibitem{KesKoiWym:J21}
M.~F. Keskin, V.~Koivunen, and H.~Wymeersch, ``\textnormal{Limited feedforward
  waveform design for OFDM dual-functional radar-communications},''
  \emph{{IEEE} Trans. Signal Process.}, vol.~69, pp. 2955--2970, Apr. 2021.

\bibitem{HuaHanXu:J24}
H.~Hua, T.~X. Han, and J.~Xu, ``\textnormal{MIMO integrated sensing and
  communication: CRB-rate tradeoff},'' \emph{{IEEE} Trans. Wireless Commun.},
  vol.~23, pp. 2839--2854, Apr. 2024.

\bibitem{RenPenSon:J24}
Z.~Ren, Y.~Peng, X.~Song, Y.~Fang, L.~Qiu, L.~Liu, D.~W.~K. Ng, and J.~Xu,
  ``\textnormal{Fundamental CRB-rate tradeoff in multi-antenna ISAC systems
  with information multicasting and multi-target sensing},'' \emph{{IEEE}
  Trans. Wireless Commun.}, vol.~23, pp. 3870--3885, Apr. 2024.

\bibitem{GuoGuWan:J25}
Y.~Guo, Y.~Gu, M.~Wang, and B.~Xia, ``\textnormal{Fundamental limits for ISAC:
  CRB-rate bound and bound-achieving input distribution},'' \emph{{IEEE} Trans.
  Wireless Commun.}, vol.~25, pp. 5605--5621, Oct. 2025.

\bibitem{YuRenPan:J25}
Z.~Yu, H.~Ren, C.~Pan, G.~Zhou, D.~Wang, C.~Yuen, and J.~Wang, ``\textnormal{A
  framework for uplink ISAC receiver designs: Performance analysis and
  algorithm development},'' \emph{arXiv e-prints, arXiv:2503.02647}, Mar. 2025.

\bibitem{MenMasPet:J25}
K.~Meng, C.~Masouros, A.~P. Petropulu, and L.~Hanzo, ``\textnormal{Cooperative
  ISAC networks: Performance analysis, scaling laws, and optimization},''
  \emph{{IEEE} Trans. Wireless Commun.}, vol.~24, pp. 877--892, Feb. 2025.

\bibitem{KwoConParWin:J21}
G.~Kwon, A.~Conti, H.~Park, and M.~Z. Win, ``\textnormal{Joint communication
  and localization in millimeter wave networks},'' \emph{{IEEE} J. Sel. Topics
  Signal Process.}, vol.~15, no.~6, pp. 1439--1454, Nov. 2021.

\bibitem{WanWan:J25}
Z.~Wang and X.~Wang, ``\textnormal{Fundamental MMSE-rate performance limits of
  integrated sensing and communication systems},'' \emph{arXiv e-prints,
  arXiv:2501.01053}, Jan. 2025.

\bibitem{SonYuXu:J26}
X.~Song, X.~Yu, J.~Xu, and D.~W.~K. Ng, ``\textnormal{CRB-rate tradeoff for
  bistatic ISAC with Gaussian information and deterministic sensing signals},''
  \emph{{IEEE} Trans. Wireless Commun.}, vol.~25, pp. 11\,768--11\,782, Feb.
  2026.

\bibitem{FodFodTel:J25}
S.~Fodor, G.~Fodor, and M.~Telek, ``\textnormal{On the trade-off between angle
  of arrival and symbol estimation in bistatic ISAC systems using unitary
  signaling},'' \emph{{IEEE} Trans. Commun.}, vol.~73, pp. 5328--5343, Jul.
  2025.

\bibitem{MaPin:J14}
J.~Ma and L.~Ping, ``\textnormal{Data-aided channel estimation in large antenna
  systems},'' \emph{{IEEE} Trans. Signal Process.}, vol.~62, no.~12, pp.
  3111--3124, Jun. 2014.

\bibitem{ZhaKamAlo:J25}
X.~Zhang, A.~Kammoun, and M.-S. Alouini, ``\textnormal{Fundamental limits via
  CRB of semi-blind channel estimation in Massive MIMO systems},'' \emph{{IEEE}
  Trans. Signal Process.}, vol.~73, pp. 3572--3587, Aug. 2025.

\bibitem{XuYuLiu:J25}
C.~Xu, X.~Yu, F.~Liu, and S.~Jin, ``\textnormal{Exploiting both pilots and data
  payloads for integrated sensing and communications},'' \emph{{IEEE} Trans.
  Wireless Commun.}, vol.~25, pp. 5573--5586, Oct. 2025.

\bibitem{KesMurMiz:C25}
M.~F. Keskin, S.~Mura, M.~Mizmizi, D.~Tagliaferri, and H.~Wymeersch,
  ``\textnormal{Bridging the gap via data-aided sensing: Can bistatic ISAC
  converge to genie performance?}'' in \emph{Proc. IEEE Radar Conf.
  (RadarConf)}, Krakow, Poland, Oct. 2025.

\bibitem{DonLiuLu:J25}
F.~Dong, F.~Liu, S.~Lu, Y.~Xiong, Q.~Zhang, Z.~Feng, and F.~Gao,
  ``\textnormal{Communication-assisted sensing in 6G networks},'' \emph{{IEEE}
  J. Sel. Areas Commun.}, vol.~43, pp. 1371--1386, Apr. 2025.

\bibitem{LiCosMit:J22}
J.~Li, M.~F.~D. Costa, and U.~Mitra, ``\textnormal{Joint localization and
  orientation estimation in millimeter-wave MIMO OFDM systems via atomic norm
  minimization},'' \emph{{IEEE} Trans. Signal Process.}, vol.~70, pp.
  4252--4264, Aug. 2022.

\bibitem{ZhaLiLiuHim:J15}
X.~Zhang, H.~Li, J.~Liu, and B.~Himed, ``\textnormal{Joint delay and doppler
  estimation for passive sensing with direct-path interference},'' \emph{{IEEE}
  Trans. Signal Process.}, vol.~64, no.~3, pp. 630--640, Oct. 2015.

\bibitem{KhaPodHaa:J21}
L.~Khamidullina, I.~Podkurkov, and M.~Haardt, ``\textnormal{Conditional and
  unconditional Cramér-Rao bounds for near-field localization in bistatic MIMO
  radar systems},'' \emph{{IEEE} Trans. Signal Process.}, vol.~69, pp.
  3220--3234, May 2021.

\bibitem{SheWin:J10a}
Y.~Shen and M.~Z. Win, ``\textnormal{Fundamental limits of wideband
  localization -- Part I: A general framework},'' \emph{{IEEE} Trans. Inf.
  Theory}, vol.~56, no.~10, pp. 4956--4980, Oct. 2010.

\bibitem{SheDaiWin:J14}
Y.~Shen, W.~Dai, and M.~Z. Win, ``Power optimization for network
  localization,'' \emph{{IEEE/ACM} Trans. Netw.}, vol.~22, no.~4, pp.
  1337--1350, Aug. 2014.

\bibitem{HasHoc:J13}
B.~Hassibi and B.~M. Hochwald, ``\textnormal{How much training is needed in
  multiple-antenna wireless links?}'' \emph{{IEEE} Trans. Inf. Theory},
  vol.~49, no.~4, pp. 951--963, Apr. 2013.

\bibitem{BoyVan:B04}
S.~Boyd and L.~Vandenberghe, \emph{\textnormal{Convex Optimization}}.\hskip 1em
  plus 0.5em minus 0.4em\relax Cambridge, U.K.: Cambridge Univ. Press, 2004.

\bibitem{RazHonLuo:J13}
M.~Razaviyayn, M.~Hong, and Z.-Q. Luo, ``A unified convergence analysis of
  block successive minimization methods for nonsmooth optimization,''
  \emph{SIAM Journal on Optimization}, vol.~23, no.~2, pp. 1126--1153, 2013.

\bibitem{ScuFacLamSon:J17}
G.~Scutari, F.~Facchinei, L.~Lampariello, and P.~Song, ``Parallel and
  distributed methods for constrained nonconvex optimization---part {I}:
  Theory,'' \emph{{IEEE} Trans. Signal Process.}, vol.~65, no.~8, pp.
  1929--1944, Apr. 2017.

\bibitem{SunBabPal:J17}
Y.~Sun, P.~Babu, and D.~P. Palomar, ``Majorization-minimization algorithms in
  signal processing, communications, and machine learning,'' \emph{{IEEE}
  Trans. Signal Process.}, vol.~65, no.~3, pp. 794--816, Feb. 2017.

\bibitem{StoNeh:J90}
P.~Stoica and A.~Nehorai, ``\textnormal{Performance study of conditional and
  unconditional direction-of-arrival estimation},'' \emph{{IEEE} Trans.
  Acoust., Speech, Signal Process.}, vol.~38, no.~10, pp. 1783--1795, Oct.
  1990.

\bibitem{BosKraYar:J16}
J.~Bosse, O.~Krasnov, and A.~Yarovoy, ``\textnormal{Direct target localization
  with an active radar network},'' \emph{Signal Process.}, vol. 125, pp.
  21--35, Jan. 2016.

\bibitem{VikHasSto:J06}
H.~Vikalo, B.~Hassibi, and P.~Stoica, ``\textnormal{Efficient joint
  maximum-likelihood channel estimation and signal detection},'' \emph{{IEEE}
  Trans. Wireless Commun.}, vol.~5, no.~7, pp. 1838--1845, Jul. 2006.

\end{thebibliography}

\end{document}